\documentclass[journal,twoside]{IEEEtran}

\usepackage{graphicx}
\usepackage{citesort}
\usepackage{amssymb}
\usepackage{amsfonts}
\usepackage{amsmath}
\usepackage{epsfig}
\usepackage{color}
\usepackage{fancybox}
\usepackage{textcomp}
\usepackage{multirow}
\usepackage{setspace}
\usepackage{psfrag}

\newtheorem{Theorem}{Theorem}

\newtheorem{Proposition}{Proposition}

\newtheorem{Lemma}{Lemma}
\newtheorem{Assumption}{Assumption}

    \def\Complex{{\rm\rule[.23ex]{.03em}{1.1ex}\kern-.3em{C}}}

    \newcommand{\be}{\begin{equation}} \newcommand{\ee}{\end{equation}}
    \newcommand{\bea}{\begin{eqnarray}} \newcommand{\eea}{\end{eqnarray}}
    \newcommand{\benum}{\begin{enumerate}} \newcommand{\eenum}{\end{enumerate}}

    \newcommand{\qg}{{\bf g}}
    \newcommand{\qh}{{\bf h}}

    \newcommand{\qq}{{\bf q}}
    
    \newcommand{\qs}{{\bf s}}

    \newcommand{\qv}{{\bf v}}
    
    \newcommand{\qx}{{\bf x}}
    \newcommand{\qy}{{\bf y}}
    \newcommand{\qz}{{\bf z}}

    \newcommand{\qA}{{\bf A}}
    \newcommand{\qB}{{\bf B}}
    \newcommand{\qC}{{\bf C}}
    \newcommand{\qD}{{\bf D}}
    \newcommand{\qE}{{\bf E}}
    
    \newcommand{\qG}{{\bf G}}
    \newcommand{\qH}{{\bf H}}
    \newcommand{\qI}{{\bf I}}

    \newcommand{\qQ}{{\bf Q}}

    \newcommand{\qT}{{\bf T}}
    \newcommand{\qU}{{\bf U}}
    \newcommand{\qV}{{\bf V}}
    \newcommand{\qW}{{\bf W}}

    \newcommand{\qPsi}{{\boldsymbol \Psi}}

    \newcommand{\qTheta}{{\boldsymbol \Theta}}
    \newcommand{\qLambda}{{\boldsymbol \Lambda}}
    
    \newcommand{\qGamma}{{\boldsymbol \Gamma}}
    
    \newcommand{\qOmega}{{\boldsymbol \Omega}}

    \newcommand{\bbR}{{\mathbb R}}
    \newcommand{\bbC}{{\mathbb C}}

    \newcommand{\calN}{{\mathcal N}}

    \newcommand{\diag}{{\sf diag}}
    
    \newcommand{\tr}{{\sf tr}}
    \newcommand{\Ex}{{\sf E}}

    \newcommand{\argmax}{\operatornamewithlimits{arg\, max}}

\begin{document}

\title{Large System Analysis of Cooperative Multi-cell Downlink Transmission via Regularized Channel Inversion with Imperfect CSIT}

\author{Jun Zhang,~\IEEEmembership{Student Member,~IEEE,}
        Chao-Kai Wen,~\IEEEmembership{Member,~IEEE,}
        Shi Jin,~\IEEEmembership{Member,~IEEE,}
        Xiqi Gao,~\IEEEmembership{Senior Member,~IEEE,}
        and~Kai-Kit Wong,~\IEEEmembership{Senior Member,~IEEE}% <-this % stops a space
\thanks{The work of J. Zhang, S. Jin, and X. Q. Gao was supported by the National Natural Science Foundation of China under Grants 60925004 and 61222102, the China High-Tech 863 Plan under Grant 2012AA01A506, the Natural Science Foundation of Jiangsu Province under Grants BK2012021, and the Supporting Program for New Century Excellent Talents in University (NCET-11-0090). The work of C.-K. Wen was supported by the National Science Council, Taiwan, under grant NSC100-2221-E-110-052-MY3.
The material in this paper was presented in part at IEEE International Conference on Communications (ICC 2012), Ottawa, Canada, June 2012. }
\thanks{J. Zhang is with the National Mobile Communications Research Laboratory, Southeast University, Nanjing 210096, P. R. China, and with the Singapore University of Technology and Design, Singapore. Email: mtzhangjun@seu.edu.cn.}
\thanks{C.-K. Wen is with Institute of Communications Engineering, National Sun Yat-sen University, Taiwan. Email: ckwen@ieee.org.}
\thanks{S. Jin and X. Q. Gao are with the National Mobile Communications Research Laboratory, Southeast University, Nanjing 210096, P. R. China, Email: \{jinshi, xqgao\}@seu.edu.cn.}
\thanks{K.-K. Wong is with Department of Electronic and Electrical Engineering, University College London, UK, Email: k.wong@ee.ucl.ac.uk.}
}
\markboth{IEEE TRANSACTIONS ON WIRELESS COMMUNICATIONS ,~Vol.~$\times$, No.~$\times$, $\times\times$~2013} {Zhang \MakeLowercase{\textit{et al.}}:} %Large System Analysis of Cooperative Multi-cell Downlink Transmission via Regularized Channel Inversion with Imperfect CSIT}

\maketitle

\begin{abstract}
In this paper, we analyze the ergodic sum-rate of a multi-cell downlink system with base station (BS) cooperation using regularized zero-forcing (RZF) precoding. Our model assumes that the channels between BSs and users have independent spatial correlations and imperfect channel state information at the transmitter (CSIT) is available. Our derivations are based on large dimensional random matrix theory (RMT) under the assumption that the numbers of antennas at the BS and users approach to infinity with some fixed ratios. In particular, a deterministic equivalent expression of the ergodic sum-rate is obtained and is instrumental in getting insight about the joint operations of BSs, which leads to an efficient method to find the asymptotic-optimal regularization parameter for the RZF. In another application, we use the deterministic channel rate to study the optimal feedback bit allocation among the BSs for maximizing the ergodic sum-rate, subject to a total number of feedback bits constraint. By inspecting the properties of the allocation, we further propose a scheme to greatly reduce the search space for optimization. Simulation results demonstrate that the ergodic sum-rates achievable by a subspace search provides comparable results to those by an exhaustive search under various typical settings.
\end{abstract}


\vspace{-.2in}
\section{Introduction}
Multi-user multiple-input multiple-output (MU-MIMO) has been well recognized as an effective means to increase capacity in the downlink \cite{Spencer-04COMMag,Gesbert-07SigMag,3GPP-LTE}. However, challenges arise in practical cellular systems where inter-cell interference remains the bottleneck limiting the achievable performance. Therefore, base station (BS) cooperation was recently proposed as a way to alleviate the issue, e.g., \cite{Karakayali06WCM,Somekh07IT,Jing08EJWCN,Somekh09IT,Boudreau09ComM,Gesbert10JSAC,Akai-Kit-10}, which is greatly motivated by the fact that BSs may be connected via high-speed backhaul links and the channel state information (CSI) and/or data and/or precoding matrices can be shared among the BSs for coordinated transmission. Such BS cooperation in the downlink can improve sum-rates and reduce outage as compared to the conventional or single-cell signal processing where the interference (often from other cells) is treated as noise.

Despite the potential, the implementation of BS cooperation faces a fundamental challenge---the availability of CSI at the transmitter side (CSIT). In frequency-division-duplex (FDD) systems,
although receivers could estimate the channel, quantize the CSI, and feed it back to the transmitter via some finite-bandwidth feedback links, CSIT will be imperfect. This is less an issue for
time-division-duplex (TDD) systems, but CSI will still be imperfect due to estimation at finite training sequence length and finite signal-to-noise ratio (SNR). As the benefits of BS cooperation
highly depend upon the quality of CSIT, recent efforts considered limited feedback models and imperfect CSIT in the design \cite{Caire-10ITA,Bhagavatula11TSP,Lee11TWC,Yuan13TWC}. The robust
beamforming based on imperfect CSIT was studied in \cite{Vucic_TSP09,Tajer_TSP11,Wang_ICASSP11,Bjornson_CAMSAP11}.

Besides the issue of availability of CSIT, other issues such as synchronization and finite capacity backhaul need significant research efforts. Nevertheless,
several testbeds for implementing the BS cooperation have recently been developed \cite{Sam-07,Irm-09,Jun-10,Irm-11} to demonstrate the feasibility of the
cooperative technique. For example, the Berlin testbed demonstrated downlink BS coordination for an FDD LTE trial system \cite{Jun-10}. Zero-forcing (ZF) precoding
based on limited CSI feedback was implemented jointly across two BSs which exchanged CSI as well as shared data over a low-latency signaling network. The BSs were
synchronized using the global positioning system (GPS).\footnote{The approaches would only be suitable for a small-scale demonstration system. A full network-wide
cooperation is yet to be seen in large-scale mobile networks. Despite this, significant gain has been reported even by forming small cooperation clusters in
large-scale networks \cite{Irm-11}.}

The success of these testbeds \cite{Sam-07,Irm-09,Jun-10,Irm-11} as well as several discussions raising in the LTE-Advanced study items \cite{R1-100820,R1-110546} has motivated us to provide further analytical results under the similar cooperative setting. Emphasis is put on providing insight into the role of the key parameters on system performance. Specifically, this paper considers a downlink system with multiple cooperative BSs serving a number of single-antenna users, in which BSs share the CSI and data via high-speed backhaul links. Also, the limited feedback and imperfect CSIT are taken into consideration. Rather than employing the ZF precoding, we consider that the BSs perform regularized ZF (RZF) \cite{Joham-02ISSSTA,Peel-05Tcom} for transmission. This is because when the channel is ill-conditioned, the achievable rates of ZF are severely compromised but RZF introduces a regularization parameter in the channel inversion to mitigate the ill-condition problem. The regularization parameter can control the amount of the introduced interference but choosing it improperly degrades the performance considerably. Ideally, one would choose the parameter to maximize the signal-to-interference plus noise ratio (SINR).

However, the system-wide SINR is a complex function of many system parameters such as channel vectors, channel-path gains, spatial correlations, imperfect CSIT, etc, which has motivated the
researchers to use some approximate SINR expressions for optimizing the regularization parameter of RZF. One promising approach to achieve this is large-system analysis by means of large
dimensional random matrix theory (RMT). Remarkably, results derived from large-system analysis also provided reliable performance predictions even for small system dimensions and at a much lower
computational cost than Monte-Carlo simulations \cite{Muharar-11ICC,Wagner12IT,Nguyen-08GLCOM,Hoydis-13JSAC,ZhangJun-13JSAC,Huang-13,He-13ICASSP}. In particular, the asymptotic-optimal
regularization parameters in the large-system limit have been obtained for independent and identically distributed (i.i.d.) channels in \cite{Nguyen-08GLCOM} and for spatially correlated
channels and imperfect CSIT in \cite{Muharar-11ICC,Wagner12IT}.

Further to the previous results, we provide a deterministic equivalent of the ergodic sum-rate for the \emph{coordinated} BS system based on large dimensional RMT. Our model takes
into account many practical factors related to the multiple BSs system. For example, channel-path gains, spatial correlations, and CSIT qualities from BSs to each user can differ. As a special
case, our contribution also complements the results of \cite{Wagner12IT} by extending the analysis to the case with multi-cell downlink coordinated systems,\footnote{For the readers'
convenience, the similar notations to those in \cite{Wagner12IT} are used.} where links have different CSIT qualities even inside a channel vector from a user to all BSs. Such extension is
nontrivial because several key manipulations for the multi-cell system with spatial correlations are required. Most importantly, our deterministic equivalent result has brought out two
fundamental applications:

\begin{itemize}
\item The deterministic equivalent sum-rate provides an efficient way to find the asymptotic-optimal regularization parameter, which by simulation results illustrates a good agreement with the optimum in terms of the ergodic sum-rates. The search of the optimal regularization parameter is a demanding process because Monte-Carlo averaging is required. Therefore, we have overcome the fundamental difficulty of applying RZF precoding in the multi-cell downlink coordinated systems.

\item In conventional single-cell processing in FDD mode, each user compares its channel vector with a predefined codebook and subsequently
 feeds back the channel index to its serving BS \cite{Au-Yeung06TWC,Jindal06IT}. Extending the technique to the case with BS cooperation would require that
 each user compares its cooperating user-BS channel pairs with a predefined codebook and then feeds back the channel indices to a single BS. The channel
 indices are then forwarded to other BSs. A \emph{single} imperfect CSIT is shared to all the BSs (see Fig.~\ref{fig:system model}).\footnote{Note that the
 CSIT configuration in this paper is different from that of \cite{Kerret-11,Kerret-13IT,Kerret-13WCM} where every BS has its \emph{own} CSIT and the CSIT of
 each BS is used to generate its own precoding locally without additional communication between BSs. This setting referred to as the \emph{distributed} CSIT
 \cite{Kerret-11,Kerret-13IT,Kerret-13WCM} is not considered in this paper.} Intuitively, those user-BS pairs with weaker channel-path gains should not require
 the same number of quantization bits as those with stronger gains. In this paper, we address the fundamental CSI feedback problem: given a total number of
 feedback bits of each user, how the feedback quantization bits be allocated among the different channel vectors.\footnote{The feedback allocation consists in
 the optimization of the bits to the different channel vectors, not to the different transmitters because every BS has the same CSIT. } This is an important
 problem, but has received little attention. Adaptive bit allocation of CSI feedback in a multi-cell system was studied in
 \cite{Bhagavatula11TSP,Lee11TWC,Yuan13TWC} but ZF/RZF precoding was not used. The feedback bit allocation problem of RZF precoding under general channel
 scenarios has not been investigated before. By inspecting the properties of feedback bit allocation, we devise a subspace method to greatly reduce the number
 of bit combinations and hence the search complexity. Computer simulations are conducted to reveal that the ergodic sum-rates by the subspace search provides
 comparable results to those by an exhaustive search for various settings.
\end{itemize}

\begin{figure}
\centering
\includegraphics[width=0.45\textwidth]{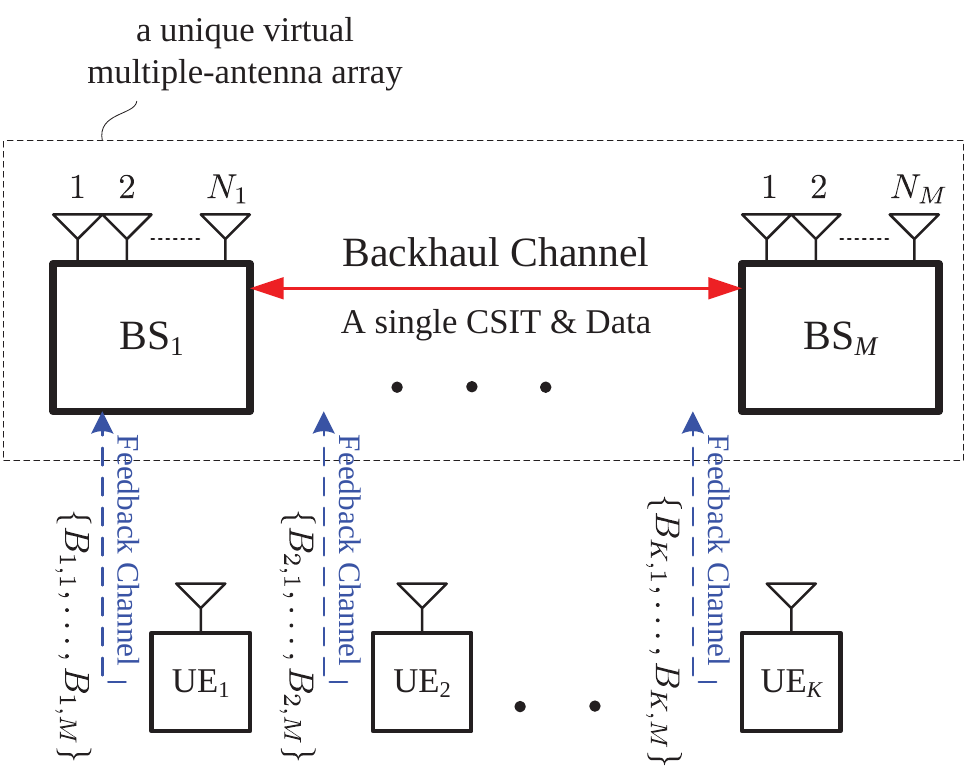}
\caption{System model of limited feedback.}\label{fig:system model}
\end{figure}

\emph{Notations}---We use uppercase and lowercase boldface letters to denote matrices and vectors, respectively. An $N \times N$ identity matrix is denoted by
${\bf{I}}_N$, while an all-zero matrix is denoted by ${\bf 0}$, and an all-one matrix by ${\bf 1}$. The superscripts $(\cdot)^{H}$, $(\cdot)^{T}$, and
$(\cdot)^{*}$ denote the conjugate-transpose, transpose, and conjugate operations, respectively. $\Ex\{\cdot\}$ returns the expectation with respect to all random
variables within the bracket. We use $[{\bf A}]_{kl}$ or the lower-case representation $A_{kl}$ to denote the ($k$,$l$)-th entry of the matrix $\bf A$, and $a_k$
denotes the $k$-th entry of the column vector $\bf{a}$. The operators $(\cdot)^{\frac{1}{2}}$, $(\cdot)^{-1}$, ${\tr}(\cdot)$ and $\det(\cdot)$ represent the
matrix principal square root, inverse, trace and determinant, respectively, and $\diag(\bf{x})$ denotes a diagonal matrix with $\bf{x}$ along its main diagonal.
The notation ``$\xrightarrow{a.s.}$'' denotes the almost sure (a.s.) convergence.

\section{Model and Problem Formulation}
\subsection{Network Model}
As shown in Fig.~\ref{fig:system model}, we consider a downlink cellular network consisting of $M$ clustered BSs with $\left\{ N_1, \ldots ,N_M \right\}$ antennas, respectively, labeled as ${\sf BS}_1,\dots,{\sf BS}_M$, serving $K$ single antenna users, labeled as ${\sf UE}_1,\dots,{\sf UE}_K$. The BSs transmit their signals jointly to users by exchanging user data and/or CSI via high-speed backhaul links. The received signal ${y_k}$ of ${\sf UE}_k$ can therefore be expressed as
\begin{equation}\label{eq:the received signal of user k}
y_k = \qh_k^H \qg_k s_k + \sum_{l=1\atop l \ne k}^K \qh_k^H \qg_l s_l + z_k,
\end{equation}
where $\qh_k^H\triangleq[\qh_{k,1}^H,\dots,\qh_{k,M}^H] \in \bbC^{1\times N}$ with $N \triangleq\sum_{i=1}^MN_i$ being the total number of transmit antennas,
$\qh_{k,i}^H \in \bbC^{1 \times N_i}$ being the channel vector between ${\sf BS}_i$ and ${\sf UE}_k$, $\qg_k\triangleq[\qg_{k,1}^T, \dots ,\qg_{k,M}^T]^T$ with
$\qg_{k,i} \in \bbC^{N_i}$ being the precoding vector between ${\sf BS}_i$ and ${\sf UE}_k$, $s_k$'s are i.i.d.~complex data symbols with zero mean and unit
variance, and $z_k$ is i.i.d.~complex Gaussian noise with zero mean and variance of $\sigma^2$.

By stacking the received signals in vector form, i.e., $\qy\triangleq[ {y_1}, \dots ,{y_K}]^T \in \bbC^{K}$, we have
\begin{equation}\label{eq:the concatenated received signal vector}
    \qy= \qH \qG \qs + \qz,
\end{equation}
where we have defined $\qH \triangleq \left[ \qh_1, \dots ,\qh_K \right]^H \in \bbC^{K\times N}$, $\qG \triangleq \left[ \qg_1, \dots, \qg_K \right]
\in\bbC^{N\times K}$, $\qs \triangleq \left[{s_1}, \dots ,{s_K} \right]^T \in \bbC^{K}$, and $\qz \triangleq \left[ {z_1}, \dots ,{z_K} \right]^T \in \bbC^{K}$.
Let $\qG_i \triangleq[\qg_{1,i}, \dots, \qg_{K,i}] \in \bbC^{N_i \times K}$ denote the overall precoding matrix of ${\sf BS}_i$. It is assumed that ${\sf BS}_i$
satisfies the following transmit power constraint:
\begin{equation}\label{eq:the per-base station transmitted power constrain}
\tr\left\{ \qG_i \qG_i^H \right\} = \tr\left\{ \qE_i \qG \qG^H \qE_i \right\} \le {N_i P},
\end{equation}
for $i = 1, \dots, M$, where ${P} > 0$ is the parameter that decides on the per-antenna power budget and
\begin{equation}
\qE_i \triangleq \diag ( {0, \dots ,0,\underbrace {1, \dots ,1}_{N_i},0, \dots ,0} ).
\end{equation}
Although the conventional sum-power constraint can achieve the better performances, the used per-BS power constraints is the more practical choice for a
cellular system. We refer to \cite{Somekh09IT} for related discussions of the two power constraints.

Due to limited angular spread and insufficient antenna spacing, the effect of spatial correlation has to be considered. For this purpose, the channel vector between ${\sf BS}_i$ and ${\sf UE}_k$ is modeled as
\begin{equation}\label{eq:the correlated channel model}
    \qh_{k,i}^H = \qx_{k,i}^H \qT_{k,i}^{\frac{1}{2}},
\end{equation}
where $\qx_{k,i}^H \in \bbC^{1 \times {N_i}}$ has i.i.d.~zero-mean entries with variance of $\frac{1}{N_i}$, and $\qT_{k,i}\in \bbC^{{N_i} \times {N_i}}$ is a deterministic nonnegative definite matrix, which characterizes the spatial correlation of the transmitted signals across the antenna elements of ${\sf BS}_i$. The channel-path gain can be included in \eqref{eq:the correlated channel model} via a scaling factor. For ease of notation, we do not introduce a new parameter but absorb the channel-path gain into $\qT_{k,i}$. In addition, to establish a deterministic equivalent result, additional assumptions on $\qx_{k,i}$'s and $\qT_{k,i}$'s are required. We will collect all the assumptions relating to the deterministic equivalent result later in Assumption \ref{Ass:1}.

It is considered that only an imperfect channel $\hat{\qh}_{k,i}$ of the true channel $\qh_{k,i}$ is available at ${\sf BS}_i$ and we model this by
\cite{Dabbagh2008Tcom,Yoo2006IT,Wagner12IT,Ding2009Tsp,Wang2006Twc}
\begin{equation}
   {\hat{\qh}}_{k,i} = \qT_{k,i}^{\frac{1}{2}}\left( \sqrt {1-\tau_{k,i}^2} \qx_{k,i} + \tau_{k,i} \qv_{k,i} \right)
   \triangleq \qT_{k,i}^{\frac{1}{2}} \hat{\qx}_{k,i},   \label{eq:an imperfect estimate H_km}
\end{equation}
where $\qv_{k,i}$ has i.i.d.~zero-mean entries with variance of $\frac{1}{N_i}$ and is independent from $\qx_{k,i}$ and $z_k$. The parameter $\tau _{k,i} \in
[0,1]$ reflects the amount of uncertainty in $\hat{\qh}_{k,i}$. In particular, $\tau _{k,i}=0$ corresponds to perfect CSIT, whereas for $\tau _{k,i}=1$ the CSIT is
completely uncorrelated to the true channel. For FDD systems, the model (\ref{eq:an imperfect estimate H_km}) reflects the imperfect channel knowledge due to the
finite-bandwidth feedback links \cite{Dabbagh2008Tcom,Wagner12IT}, whereas for TDD systems, the model (\ref{eq:an imperfect estimate H_km}) reflects the imperfect
channel estimation due to finite training sequence length \cite{Yoo2006IT,Ding2009Tsp,Wang2006Twc}. We find it useful to define
\begin{equation} \label{eq:accuracyOfChannel}
\psi_{k,i} \triangleq \sqrt{1-\tau_{k,i}^2}
\end{equation}
and $\hat{\qh}_k \triangleq \left[ \hat{\qh}_{k,1}^H, \ldots ,\hat{\qh}_{k,M}^H \right]^H$, for $k = 1,\dots,K$.

In this paper, we consider the RZF precoding \cite{Joham-02ISSSTA,Peel-05Tcom} so that
\begin{equation}
    \qG = \xi \left( \hat{\qH}^H \hat{\qH} + \alpha \qI_N \right)^{-1} \hat{\qH}^H
        \triangleq \xi \hat{\qW} \hat{\qH}^H,  \label{eq:the RZF precoding}
\end{equation}
where $\hat{\qH}^H \triangleq [ \hat{\qh}_1,\hat{\qh}_2, \ldots ,\hat{\qh}_K ] \in \bbC^{N\times K}$ denotes the channel estimate available at the BSs, $\xi$ is a
normalization scalar to fulfill the per-BS transmit power constraint \eqref{eq:the per-base station transmitted power constrain}, $\alpha > 0$ represents the
regularization scalar, and $\hat{\qW} \triangleq ( \hat{\qH}^H \hat{\qH} + \alpha \qI_N )^{-1}$. Such precoding is a practical linear precoding scheme to control
inter-user interference and increase the sum rate by the regularization parameter $\alpha$ \cite{Peel-05Tcom}. It covers two well-known precoding with $\alpha=0$
being the ZF precoding, and $\alpha\rightarrow\infty$ giving the matched-filter precoding.

From \eqref{eq:the per-base station transmitted power constrain}, we also define
\begin{align*}
\xi_i^2 & \triangleq \frac{N_i  P }{\tr\left\{ \qG_i \qG_i^H \right\} }=\frac{\frac{N_i}{N} P }{\frac{1}{N} \tr\left( \qE_i \hat{\qW} \hat{\qH}^H \hat{\qH} \hat{\qW}^H \qE_i \right)}  \nonumber\\
        & \triangleq \frac{\frac{N_i}{N} P }{\Phi_i (\alpha )},
\end{align*}
for $i = 1, \dots, M$. To satisfy \eqref{eq:the per-base station transmitted power constrain}, we set $\xi^2  = \mathop {\min }\limits_i \left\{\xi_i^2
\right\}$.\footnote{ This choice may degrade the system performance. However, the performance loss can be minimized though user selection schemes. The related
effect will be discussed later in Section III.} Then, the SINR of ${\sf UE}_k$ is given by
\begin{align}
\gamma_{k}&=\frac{| \qh_k^H \hat{\qW} \hat{\qh}_k|^2 }{ \qh_k^H \hat{\qW} \hat{\qH}_{[k]}^H \hat{\qH}_{[k]} \hat{\qW}^H \qh_k + \frac{\sigma^2}{\xi^2} }  \nonumber\\
          &=\frac{| \qh_k^H \hat{\qW} \hat{\qh}_k|^2 }{ \qh_k^H \hat{\qW} \hat{\qH}_{[k]}^H \hat{\qH}_{[k]} \hat{\qW}^H \qh_k + \nu },       \label{eq:SINR}
\end{align}
where $\hat{\qH}_{[k]} \triangleq  [ \hat{\qh}_1, \dots ,\hat{\qh}_{k-1},\hat{\qh}_{k+1}, \dots ,\hat{\qh}_K ]^H \in \bbC^{N\times(K-1)}$, $\nu \triangleq\max_i
\frac{ N \Phi_i (\alpha)}{N_i\rho}$ and $\rho \triangleq \frac{ P }{\sigma^2}$. Under this model, the ergodic sum-rate can be defined as
\begin{equation}\label{eq:the ergodic sum rate}
R_{\rm{sum}}\triangleq\sum_{k=1}^{K} \Ex_{\hat{\qH}} \left\{ \log \left( 1 + \gamma_{k}\right) \right\}.
\end{equation}

\subsection{A Fundamental Problem}
The SINR, $\gamma_{k}$, in \eqref{eq:SINR} is a function of the regularization parameter $\alpha$. It has been well understood in the literature that adopting an
improper regularization parameter would degrade performance significantly \cite{Peel-05Tcom,Wagner12IT}. As a consequence, to use RZF precoding effectively, it is
important to obtain an appropriate value of $\alpha$ for best performance. In this paper, we are interested in finding the optimal regularization parameter that
maximizes the ergodic sum-rate \eqref{eq:the ergodic sum rate}. That is, we have
\begin{equation}\label{eq:optimization of the regularization parameter_exp}
 \alpha^{\rm opt} =  \argmax_{\alpha > 0} R_{\rm{sum}}.
\end{equation}
The problem above does not admit a simple closed-form solution and the solution needs to be computed via a one-dimensional linear search, such as the golden section search \cite[Chapter 7]{Chong-08BOOK}. However, Monte-Carlo averaging over the channels is required to evaluate the ergodic sum-rate \eqref{eq:the ergodic sum rate} for each choice of $\alpha$, which, unfortunately, makes the overall computational complexity prohibitive and this drawback hinders the development of RZF precoding in the multi-cell downlink channel. To tackle this problem, we resort to an asymptotic expression of \eqref{eq:the ergodic sum rate} in the large-system regime which we derive in the next section.

\section{Main Theoretical Results}
In this section, we present a deterministic equivalent of the SINR of the RZF precoding system. To do so, we consider the large-system regime where $N_i$'s and $K$
approach to infinity at fixed ratios $\{\beta_i = N_i/K\}_{\forall i}$ such that $0 < \lim \inf \beta_i \le \lim \sup \beta_i  < \infty$. For brevity, we simply
use $\calN \rightarrow \infty$ to represent the quantity in such limit. In addition, we impose the following assumptions in our derivations.
\begin{Assumption} \label{Ass:1}
For the channel ${\hat{\qh}}_{k,i}$ in \eqref{eq:an imperfect estimate H_km}, we have the following hypotheses for $k \in \{1,\dots,K\}$:
\begin{enumerate}
\item[1)] $\qx_k\triangleq[\qx_{k,1}^H \cdots \qx_{k,M}^H]^H \in \bbC^{N}$ and ${\bf x}_{k,i}$ has i.i.d.~zero-mean entries with variance of $\frac{1}{N_i}$ and finite $8$-th order moment.
\item[2)] $\qv_k = [\qv_{k,1}^H \cdots \qv_{k,M}^H]^H \in \bbC^{N}$ has the same statistical properties as $\qx_k$, but they are independent.
\item[3)] $\qT_k = \diag \left( \qT_{k,1}, \qT_{k,2}, \dots , \qT_{k,M} \right) \in \bbC^{N\times N}$ with $\qT_{k,i} \in \bbC^{N_i\times N_i}$'s being nonnegative definite matrices with uniformly bounded spectral norm.
\end{enumerate}
\end{Assumption}

\begin{Theorem}\label{Th: 2}
Under Assumption \ref{Ass:1}, as $\calN \rightarrow \infty$, we have $\gamma_{k} - \overline{\gamma}_{k}   \xrightarrow{a.s.} 0, ~~\mbox{for } k = 1,\dots,K$,
where
\begin{equation}\label{eq:gamma deterministic equivalent}
\overline{\gamma}_{k}   = \frac{\left(\overline{u}^{(2)}_k \right)^2}{\left( 1+\overline{u}^{(1)}_k  \right)^2\left(\overline{u}_k+\overline{\nu} \right)},
\end{equation}
with
\begin{subequations} \label{eq:u_all}
\begin{align}
    \overline{u}^{(1)}_k &= \sum_{i=1}^{M} \frac{1}{N_i} \tr \left(\qT_{k,i} \qPsi_i \right),
    \label{eq: u1_circ} \\
    \overline{u}^{(2)}_k &= \sum_{i=1}^{M} \frac{\psi_{k,i}}{N_i} \tr \left(\qT_{k,i} \qPsi_i \right),
    \label{eq: u2_circ} \\
    \overline{\nu} =& \mathop {\max }\limits_i \frac{1}{N_{i} \rho} \biggl( \tr \qPsi_i   - \alpha \tr \qPsi_i^2  + \alpha  \sum _{k=1}^{K} \dot{c}_{k,i}   \tr \left(\qPsi_i  \qT_{k,i} \qPsi_i \right) \biggr),
    \label{eq: nu_circ} \\
    \overline{u}_k &= \overline{u}^{(1)}_k  - \alpha \overline{\dot u}_k^{(1)} \nonumber\\
                    &~~- \frac{2\overline{u}^{(2)}_k \left(\overline{u}^{(2)}_k - \alpha \overline{\dot u}_k^{(2)} \right)}{1+\overline{u}^{(1)}_k }+ \frac{\left(\overline{u}^{(2)}_k \right)^2 \left(\overline{u}^{(1)}_k  - \alpha \overline{\dot u}_k^{(1)}  \right)}{\left(1+\overline{u}^{(1)}_k  \right)^2},
    \label{eq: u_circ} \\
    \overline{\dot u}_k^{(1)}& = \sum_{i=1}^{M}{\frac{1}{N_i} \left[ \tr{\left(\qT_{k,i} \qPsi_i^2 \right)} - \sum_{l=1}^{K}{ \dot{c}_{l,i}\tr{\left(\qT_{k,i} \qPsi_i  \qT_{l,i} \qPsi_i \right)}} \right]},
    \label{eq: dotu1_circ} \\
    \overline{\dot u}_k^{(2)} & = \sum_{i=1}^{M}{\frac{\psi_{k,i}}{N_i} \left[ \tr{\left(\qT_{k,i} \qPsi_i^2 \right)} - \sum_{l=1}^{K}{ \dot{c}_{l,i}\tr{\left(\qT_{k,i} \qPsi_i  \qT_{l,i} \qPsi_i \right)}} \right]}. \label{eq: dotu2_circ}
\end{align}
\end{subequations}
In \eqref{eq:u_all}, $e_{k,i}$'s are the unique solution of the following $K \times M$ equations
\begin{align}
\qPsi_{i}  &= \left( \frac{1}{N_i}\sum_{k=1}^{K}{\frac{1}{\sum_{m=1}^{M} e_{k,m} + 1} \qT_{k,i} } + \alpha \qI_{N_i} \right)^{ -1}, \label{eq: Psi i} \\
e_{k,i}  &= \frac{1}{N_i}\tr \left(\qT_{k,i} \qPsi_i \right), \label{eq: e km}
\end{align}
for $k = 1,\dots,K$ and $i = 1, \dots, M$. In addition, $\dot{\qC}   = [\dot{c}_{k,i}  ] \in \bbC^{K \times M}$ is a solution to the following linear equation:
\begin{equation}\label{eq: Theta}
\qTheta  {\sf vec}(\dot{\qC})  = {\sf vec}(\qGamma),
\end{equation}
where $\qTheta = [\qTheta_{ik}] \in \bbC^{M K\times M K}$ and $\qGamma \in \bbC^{K \times M}$, with
\begin{subequations}
\begin{align}
[\qTheta_{ik}]_{jl} &= \left\{\begin{aligned}
\frac{-1}{N_{i}} \frac{1}{N_{j}} \frac{1}{\left(\sum_{m = 1}^M e_{k,m} +1 \right)^2} \tr{\left(\qT_{k,j} \qPsi_{j} \qT_{l,j} \qPsi_{j} \right)}, & \\ \mbox{for }(i,j)\neq (l,k);& \\
1-\frac{1}{N_{i}^2}  \frac{1}{\left(\sum_{m = 1}^M e_{k,m} +1 \right)^2} \tr{\left(\left(\qT_{k,i} \qPsi_{i} \right)^2\right)},  ~~~& \\ \mbox{for }(i,j)= (l,k),&
\end{aligned}\right.\label{eq: qTheta_ik}\\
[\qGamma]_{k,i}& = -\frac{1}{N_{i}} \sum\limits_{j = 1}^M \frac{1}{N_{j}} \frac{1}{\left(\sum_{m = 1}^M e_{k,m} +1\right)^2}\tr \left(\qT_{k,j} \qPsi^2_{j} \right).\label{eq: eta_k}
\end{align}
\end{subequations}
\end{Theorem}

\begin{proof}
See Appendix \ref{Appendix: Proof of Theorem 2}.
\end{proof}

An intuitive application of Theorem \ref{Th: 2} is that $\gamma_{k}$ can be approximated by its deterministic equivalent $\overline{\gamma}_{k}$. Given statistical channel knowledge, i.e.,
($\qT_{k,i}$'s, $\tau_{k,i}$'s, and $\sigma^2$), the SINR of ${\sf UE}_k$ can be approximated by Theorem \ref{Th: 2}, without knowing the actual channel realization. Using the definition of the
deterministic equivalent \cite[Definition 6.1]{Couillet-11BOOK} and the continuous mapping theorem \cite{Billingsley-95BOOK}, we have $\log \left( 1 + \gamma_{k} \right) - \log \left( 1 +
\overline{\gamma}_{k} \right) \rightarrow 0$ almost surely as $\calN \rightarrow\infty$. An approximation $\overline{R}_{\rm{sum}}$ of the ergodic sum-rate $R_{\rm{sum}}$ in \eqref{eq:the
ergodic sum rate} is obtained by replacing the instantaneous SINR $\gamma_{k}$ with its large system approximation $\overline{\gamma}_{k}$, i.e.,
\begin{equation*}
\overline{R}_{\rm{sum}} = \sum_{k=1}^{K} \log \left( 1 + \overline{\gamma}_{k} \right).
\end{equation*}
It follows that \cite{Wagner12IT}
\begin{equation*}
\frac{1}{K} \left( R_{\rm{sum}} - \overline{R}_{\rm{sum}} \right)  \xrightarrow{\calN \rightarrow\infty}0
\end{equation*}
holds true almost surely.

To derive the deterministic equivalent $\overline{\gamma}_{k}$, we first need to obtain the fixed-point solutions ${e_{k,i}}$'s which can be easily solved by iteratively solving the equations \eqref{eq: e km} and \eqref{eq: Psi i}. With ${e_{k,i}}$'s as well as $\qPsi_{i}$'s, all $\dot{c}_{k,i}$'s can be obtained by solving the simple linear equation in \eqref{eq: Theta}. Finally, substituting ${e_{k,i}}$'s, $\qPsi_{i}$'s, and $\dot{c}_{k,i}$'s into \eqref{eq:u_all}, we then yield all the required parameters and as a result get the final estimate.

It should be emphasized here that our model takes into account the BS cooperation as well as the effect of imperfect CSIT. Therefore, Theorem \ref{Th: 2} is general and can be interpreted
as a unified formula that encompasses many known results, such as \cite[Theorem 1]{Wagner12IT}. Specifically, in contrast to \cite{Wagner12IT}, the new analytical result enables us to deal
with the more general case where links have different CSIT qualities even inside a channel vector from a user to all BSs. To have a better understanding on the expression of the deterministic
equivalent, Theorem \ref{Th: 2} is applied to the following two special cases.

First, if CSIT is perfect, then in this case, from \eqref{eq:accuracyOfChannel}, we have $\tau_{k,i} = 0$ and therefore $\psi_{k,i} = 1$. Then, from
\eqref{eq:u_all}, one can easily obtain $\overline{u}^{(1)}_k = \overline{u}^{(2)}_k$ and $\overline{\dot u}_k^{(1)} = \overline{\dot u}_k^{(2)}$. As such,
$\overline{u}_k$ in \eqref{eq: u_circ} can be simplified as $\overline{u}_k = \frac{\overline{u}^{(1)}_k - \alpha \overline{\dot
u}_k^{(1)}}{(1+\overline{u}^{(1)}_k)^2}$. Substituting this result into \eqref{eq:gamma deterministic equivalent}, we finally get
\begin{equation}\label{eq:gamma deterministic equivalent_perfectCSIT}
\overline{\gamma}_{k, {\sf Perfect}} = \frac{\left(\overline{u}^{(1)}_k \right)^2}{ \left(\overline{u}^{(1)}_k - \alpha \overline{\dot u}_k^{(1)} \right) + \left( 1+\overline{u}^{(1)}_k  \right)^2\overline{\nu} }.
\end{equation}
Now, if both $\alpha$ and $\sigma^2$ (hence $\overline\nu$) are close to zero, we have $\overline{\gamma}_{k, {\sf Perfect}} \approx\overline{u}^{(1)}_k$.\footnote{The approximation can be obtained from (\ref{eq:gamma deterministic equivalent_perfectCSIT}) by naively letting $\alpha$ and $\overline\nu$ be zero. The approximation helps us better understand the complex expression although we do not have a formal argument. } Recalling from the form of $\overline{u}^{(1)}_k$ in \eqref{eq: u1_circ}, we know that it comprises of the sum of $\tr \left(\qT_{k,i} \qPsi_i \right)$ and therefore, $\tr \left(\qT_{k,i} \qPsi_i \right)$ can be viewed as the equivalent channel gain contributed by ${\sf BS}_i$ to ${\sf UE}_k$. As expected, the multi-cell cooperation appears to provide a combining-like gain. In contrast to several channel parameters such as $\rho$ and $\qT_{k,i}$, the factor $\tr \left(\qT_{k,i} \qPsi_i \right)$ not only reflects the effect due to these channel parameters but also the inter-user interference. In other words, $\tr\left(\qT_{k,i} \qPsi_i \right)$ serves as a good indicator to illustrate the channel gain from ${\sf BS}_i$ to ${\sf UE}_k$ and will later be used in our design for the optimal feedback bit allocation.

In another special case, when $M=1$, Theorem \ref{Th: 2} agrees with the results in \cite[Theorem 1]{Wagner12IT}. Although the differences between the cases with $M > 1$ and $M = 1$ appear on each factor of \eqref{eq:u_all}, we can observe $\qTheta$ in \eqref{eq: qTheta_ik} that multi-cell cooperation appears to involve correlations between the BSs.

Besides the special cases, another interesting observation from the deterministic equivalent expression is the effect due to the normalization scalar $\xi$. Recall
that we set $\xi^2 = \mathop {\min }\limits_i \left\{\xi_i^2 \right\}$. It seems that the joint RZF precoding over cooperating BSs will enforce the power of each
BS by scaling the total transmit power over all BSs according to the worst-conditioned BS and this may degrade the system performance significantly. However, via
the deterministic equivalent, we show that the effect is not a serious issue. The corresponding effect of $\xi$ on the SINR is through $\overline{\nu}$ in
(\ref{eq: nu_circ}). As can be seen, $\overline{\nu}$ is associated with $\qPsi_i$ while $\qPsi_i$ is determined by $\{ \qT_{k,i} \}_{k=1,\ldots,K}$. Since $\{
\qT_{k,i} \}$ are random, there is no specific reason why any specific BS would undergo very worst condition. It should be particularly noted that the
channel-path gain does not affect much the power normalization factor while the ill-condition of the channel does. The ill-condition of the channel would result
from users with similar channel responses. Clearly, if applying proper user selection schemes, the BSs' conditions will not diverge greatly.

\section{Applications and Simulations}
As described previously, the deterministic equivalent result in Theorem \ref{Th: 2} provides an efficient technique to estimate the ergodic sum-rate for the RZF precoding system. In Section IV-A, computer simulations are provided to evaluate the accuracy of the deterministic equivalent sum-rate $\overline{R}_{\rm{sum}}$. We will see that the deterministic equivalent is useful even for systems with finite numbers of antennas. In Section IV-B, the deterministic equivalent result will be used to determine a proper regularization parameter, which addresses the fundamental problem for RZF precoding. Finally in Section IV-C, we proceed to answer the fundamental question: How should the feedback bits be allocated among the BSs?

\subsection{Accuracy of the Deterministic Equivalent Sum-rate}
In this subsection, we present numerical results to confirm our analytical results under various settings. We will compare the analytical results \eqref{eq:gamma deterministic equivalent} in Theorem \ref{Th: 2} and Monte-Carlo simulation results obtained from averaging over a large number of i.i.d.~Rayleigh block fading channels.

Fig.~\ref{fig:fig2} shows the results of the ergodic sum-rate and the deterministic equivalent rate with $M=4, N_1= \cdots = N_4=8, K=32$ and $\qT_{k,i} = \qI$ for the following five cases: 1) $\{\tau_{k,i}^2 =0\}_{\forall k,i}$, 2) $\{\tau_{k,i}^2 =0.1\}_{\forall k,i}$, 3) $\{\tau_{k,i}^2 =0.2\}_{\forall k,i}$, 4) $\{\tau_{k,i}^2 = 0.3\}_{\forall k,i}$, and 5) $\{\tau_{k,1}^2 = 0, \tau_{k,2}^2 = 0.1, \tau_{k,3}^2 = 0.2, \tau_{k,4}^2 = 0.3\}_{\forall k}$, respectively. In this figure, we compare the cases where the regularization parameter $\alpha$ is obtained by \eqref{eq: opt alpha} with those by $\alpha = \frac{1}{M \rho \beta}$. The optimality of the regularization parameters will be discussed later in the next subsection. We see that the deterministic equivalent is accurate even for systems with finite numbers of antennas.

\begin{figure}
\centering
\includegraphics[width=0.485\textwidth]{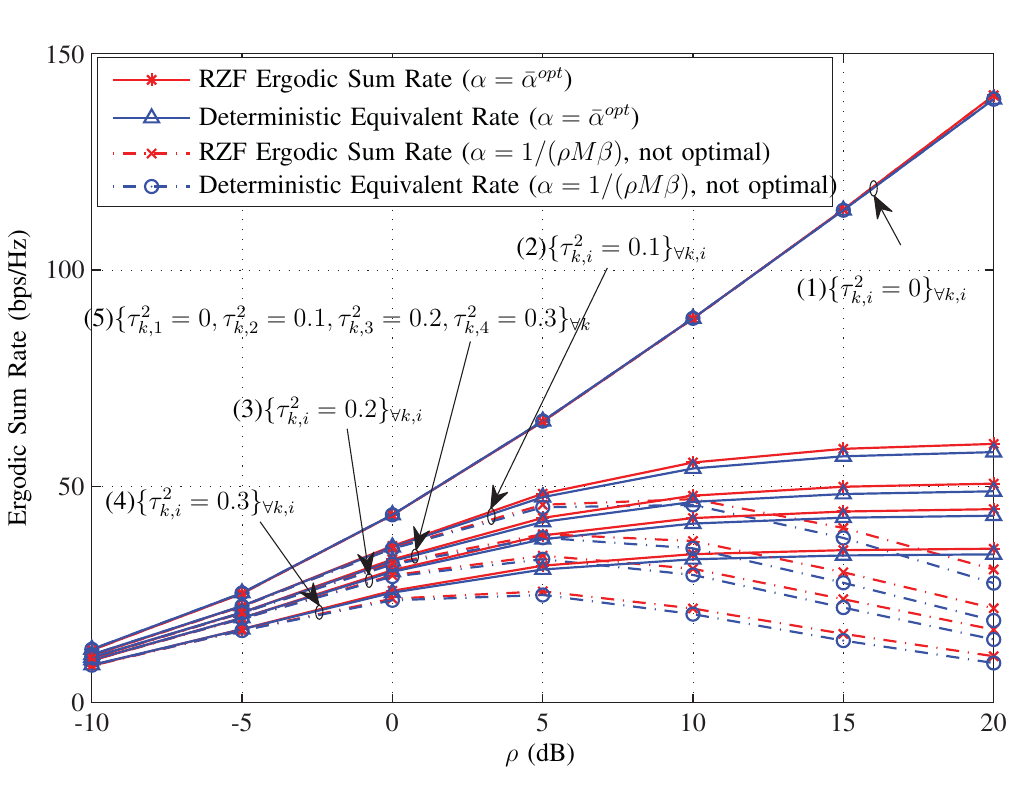}
\caption{Ergodic sum-rate and the deterministic equivalent results with $M=4, N_1= \cdots = N_4=8, K=32$ and $\qT_{k,i} = \qI$.}\label{fig:fig2}
\end{figure}

Next, we examine the accuracy of the analytical result \eqref{eq:gamma deterministic equivalent} for general settings. When ``$\tau_{k,i}^2 = {\rm rand}$'', $\tau_{k,i}^2$ is a uniform random number between $0$ and $1$, whereas $\qT_{k,i} \neq \qI$ indicates that the spatial correlation is an arbitrary pattern. The errors between the ergodic sum-rate and the deterministic equivalent rate will be different depending on $\tau_{k,i}^2$ and $\qT_{k,i}$. Therefore, the Monte-Carlo simulations are averaged over not only a large number of i.i.d.~Rayleigh block fading channels but also a large number of $\tau_{k,i}^2$ and $\qT_{k,i}$. Fig.~\ref{fig:Accuracy_M2DiffTauDiffNt} shows the relative error $\frac{R_{\rm{sum}} - \overline{R}_{\rm{sum}}}{R_{\rm{sum}}}$ with increasing number of antennas when $\tau_{k,i} = 0$ or $\tau_{k,i} = {\rm rand}$. As expected, the deterministic equivalent rate becomes more accurate if the number of antennas increases. Also, it is observed that the convergence rate becomes slow with imperfect CSIT. We can conclude from Fig.~\ref{fig:fig2} and \ref{fig:Accuracy_M2DiffTauDiffNt} that the deterministic equivalent is accurate even with finite numbers of antennas.

\begin{figure}
\centering
\includegraphics[width=0.485\textwidth]{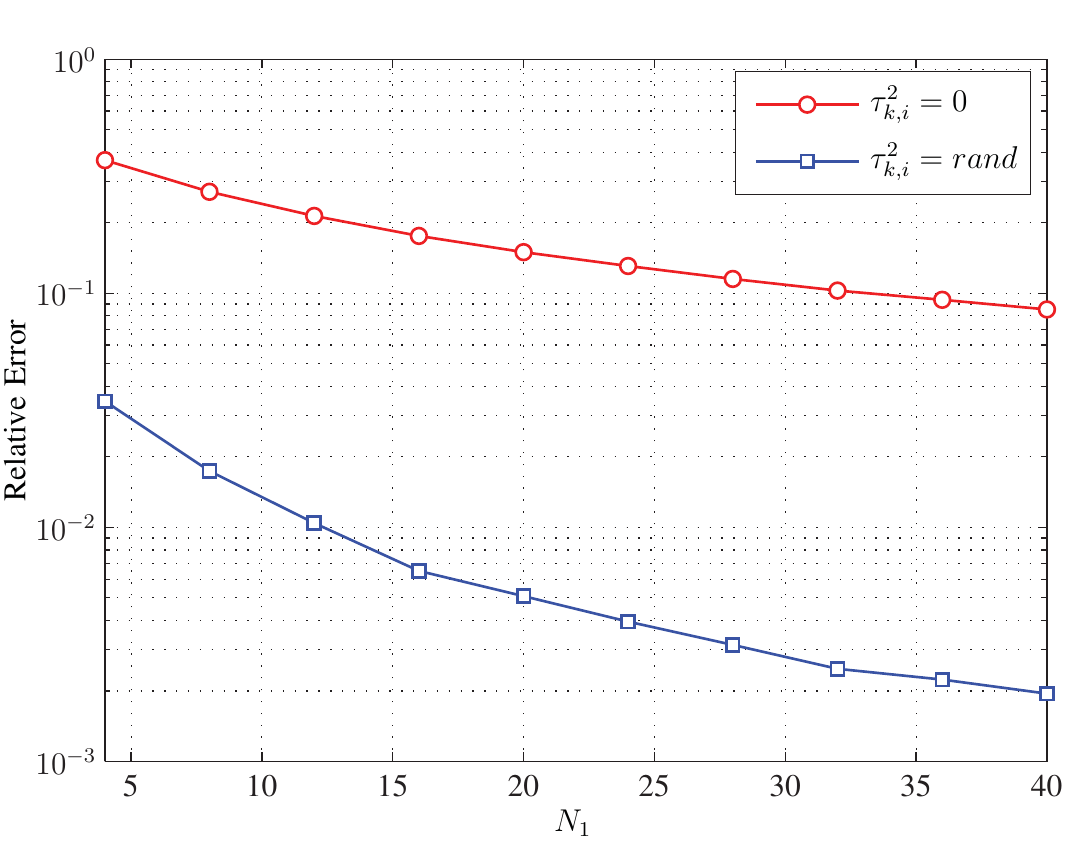}
\caption{Relative error versus the system dimension with $M=2$, $N_1=N_2$, $K=2 N_1$, $\rho = 20$dB, and $\qT_{k,i} \neq \qI$. The regularization parameters $\alpha$'s
are obtained by \eqref{eq: optimization of the regularization parameter}.}\label{fig:Accuracy_M2DiffTauDiffNt}
\end{figure}

\subsection{Sum-rate Maximizing Regularization}
Because of the high accuracy of the deterministic equivalent sum-rate, it can be used to determine the regularization parameter. Here, we focus on this particular optimization for maximizing the deterministic equivalent sum-rate:
\begin{equation}\label{eq: optimization of the regularization parameter}
\overline{\alpha}^{\rm opt} = \argmax_{\alpha > 0} \sum^K_{k=1}\log(1+\overline{\gamma}_{k}).
\end{equation}
Same as before, the optimal solution $\overline{\alpha}^{\rm opt}$ does not permit a closed-form solution. However, this time the optimal solution can be computed very efficiently via the golden section search \cite[Chapter 7]{Chong-08BOOK} without the need of Monte-Carlo averaging because $\overline{\gamma}_{k}$ is deterministic. For a special case, we obtain a solution for the optimization of the regularization parameter in the following proposition.
\begin{Proposition}\label{proposition: opt alpha}
Let $N_i = N_1 ~\forall i$, $\tau_{k,i}=\tau_i \in [0,1)$, $\qT_{k,i} = \qT$, $\forall k,i$, and denote $\beta = \frac{N_1}{K}$. The optimal $\overline{\alpha}^{\rm opt}$ in \eqref{eq: optimization of the regularization parameter} is given as a positive solution to the fixed-point equation \eqref{eq: opt alpha_1}, shown at the top of next page,
where $e_1(\alpha) = \frac{1}{N_1}\tr \left(\qT \qPsi\right), e_2(\alpha)= \frac{1}{N_1}\tr \left(\qT \qPsi^2\right), e_3(\alpha)= \frac{1}{N_1}\tr \left(\qT \qPsi\right)^2, e_4(\alpha)= \frac{1}{N_1}\tr \left(\qT \qPsi^3\right), e_5(\alpha) = \frac{1}{N_1}\tr \left(\qT^2 \qPsi^3\right)$, and
\begin{figure*}[!t]
\begin{equation}\label{eq: opt alpha_1}
\overline{\alpha}^{\rm opt} = \frac{ \left(1+\eta(\overline{\alpha}^{\rm opt}) \right) e_2(\overline{\alpha}^{\rm opt})+ M \rho \left(1 - \psi^2\right)e_3(\overline{\alpha}^{\rm opt})}{ M\beta \rho e_2(\overline{\alpha}^{\rm opt}) \left( \left(1+\eta(\overline{\alpha}^{\rm opt})\right)\psi^2 + \left(1+M e_1(\overline{\alpha}^{\rm opt})\right)^2\left(1-\psi^2\right)\eta(\overline{\alpha}^{\rm opt}) \right)}
\end{equation}
%\hrulefill
\end{figure*}
\begin{align}
\psi &=  \frac{1}{M}\sum^M_{i=1} \sqrt{1-\tau^2_i} ~~\in [0,\,1], \label{eq:delta} \\
\eta(\alpha)&=\frac{e_3e_4-e_5e_2}{Me_2^2e_3},      \label{eq:eta}          \\
\qPsi  &= \left( {\frac{1}{\beta \left(M e_1 + 1\right)} \qT } +\alpha \qI_{N_1} \right)^{ -1}.
\end{align}
\end{Proposition}

\begin{proof}
The optimization of the regularization parameter $\alpha$ satisfies
\begin{equation*}
\sum^K_{k=1} \frac{1}{1+\overline{\gamma}_{k}}\frac{\partial \overline{\gamma}_{k}}{\partial \alpha} =0.
\end{equation*}
Substituting $N_i = N_1$, $\tau_{k,i}=\tau_i$, $\qT_{k,i} = \qT$, $\forall k,i$ into Theorem \ref{Th: 2}, we have
\begin{equation}\label{eq: opt gamma}
\overline{\gamma}_{k} =\frac{ M^2 \psi^2 \rho e_1 ( Me_3 + \alpha \beta e_2\left( {Me_1 + 1} \right)^2 )} { ( \left(1 +  Me_1\right)^2 \left(1 - \psi^2\right)+ \psi^2)Me_3\rho +\left( Me_1+1 \right)^2 e_2},
\end{equation}
where $\psi$ is given by \eqref{eq:delta}. Differentiating both sides of \eqref{eq: opt gamma} with respect to $\alpha$ and using the fact that
\begin{align*}
e_1&=\frac{e_3}{\beta(Me_1+1)}+\alpha e_2,  \\
\frac{\partial e_1}{\partial \alpha} &=-\frac{\beta(Me_1+1)^2e_2}{\beta(Me_1+1)^2-Me_3},
\end{align*}
we obtain
\begin{align*}
\frac{\partial \overline{\gamma}_{k}}{\partial \alpha}=&\Delta \bigg( \left( \psi^2(1+\eta)+(1+Me_1)^2\left(1-\psi^2\right)\eta\right) \alpha \nonumber \\
                        & -  \frac{ \left(1+\eta \right) e_2+ M \rho \left(1 - \psi^2\right)e_3} { M\beta \rho e_2 }  \bigg),
\end{align*}
where $\eta$ is given by \eqref{eq:eta} and $\Delta$ is given by \eqref{eq:Delta}, shown at the top of next page.
\begin{figure*}[!t]
\begin{equation} \label{eq:Delta}
\Delta = \frac{2M^2\beta\rho \overline{\gamma}_{k} e^2_1 e_3 (1+Me_1)}{\left( e_3+\alpha\beta (1+Me_1)^2e_2\right)\left( \left(\psi^2+(1+Me_1)^2\left(1-\psi^2\right)\right) \rho Me_3+ (1+Me_1)^2e_2\right)} \frac{\partial e_1}{\partial \alpha}
\end{equation}
\hrulefill
\end{figure*}
Due to the fact that $\Delta \neq 0$, the optimal $\alpha$ satisfies \eqref{eq: opt alpha_1}.
\end{proof}

A large number of simulation results strongly suggest that the fixed-point equation \eqref{eq: opt alpha_1} has a unique fixed-point. However, it is difficult to
obtain a strict proof of the existence and uniqueness, and is still an open challenge. The results in Proposition \ref{proposition: opt alpha} can be used to
analyze a number of interesting special cases, which we provide as follows:
\begin{itemize}
  \item {\em Uncorrelated channel}---When $\qT_{k,i} = \qI_{N_1}$, we have $e_2=e_3=e_1^2$ and $e_4=e_5=e_1^3$. Thus, a closed-form solution for the optimization of the regularization parameter is given by
        \begin{equation}\label{eq: opt alpha}
        \overline{\alpha}^{\rm opt} = \frac{\frac{1}{M} + \rho \left(1 - \psi^2\right)}{ \beta \rho \psi^2}.
        \end{equation}
  \item {\em Single cell}---When $M=1$, the optimal $\overline{\alpha}^{\rm opt}$ satisfies \eqref{eq: opt alpha_3}, shown at the top of next page.
        \begin{figure*}[!t]
        \begin{equation}\label{eq: opt alpha_3}
        \overline{\alpha}^{\rm opt} = \frac{ \left(1+\eta(\overline{\alpha}^{\rm opt}) \right) e_2(\overline{\alpha}^{\rm opt})+ \rho \left(1 - \psi^2\right)e_3(\overline{\alpha}^{\rm opt})}{ \beta \rho e_2(\overline{\alpha}^{\rm opt}) \left( \left(1+\eta(\overline{\alpha}^{\rm opt})\right)\psi^2 + \left(1+ e_1(\overline{\alpha}^{\rm opt})\right)^2\left(1-\psi^2\right)\eta(\overline{\alpha}^{\rm opt})\right)}
        \end{equation}
        \hrulefill
        \end{figure*}
        This agrees with the results in \cite{Wagner12IT}.
  \item {\em Perfect CSIT}---When $\tau_{k,i}= 0$, we have $\psi =1$. The optimal $\overline{\alpha}^{\rm opt}$ degenerates to
        \begin{equation}\label{eq: opt alpha_2}
        \overline{\alpha}^{\rm opt} = \frac{ 1}{ M \rho\beta }.
        \end{equation}
\end{itemize}

Using \eqref{eq: opt alpha} or \eqref{eq: opt alpha_2}, we can observe how multi-cell cooperation affects the regularization parameter. The observation is new due to our analytical result in Proposition \ref{proposition: opt alpha}. In particular, we compare $\psi_{k,i}$ in
\eqref{eq:accuracyOfChannel} with $\psi$ in \eqref{eq:delta}. It appears that multi-cell cooperation provides the average effect over the channel uncertainty. On the other hand, if the channel
uncertainty $\tau_{k,i}$ is unknown at the transmitters (hence assumed to be zero), we will have $\alpha = \frac{1}{M \rho \beta}$, which is the benchmark considered in Fig.~\ref{fig:fig2}.
Clearly, significant performance loss is observed if the imperfect CSIT is not addressed.

In Fig.~\ref{fig:optAlpha_MonteVsLarge}, we compare the ergodic sum-rate results for various regularization parameters. Here, $\qT_{k,i}$ is generated from an arbitrary pattern and $\tau_{k,i}^2$ is obtained from a uniform random number between $0$ and $1$. The best result of $\alpha^{\rm opt}$ is obtained by maximizing the ergodic sum-rate which is calculated by Monte-Carlo averaging over $10^4$ independent trials. Such direct maximization clearly leads to very high computational cost. As we can see, $\overline{\alpha}^{\rm opt}$ provides indistinguishable results to that achieved by $\alpha^{\rm opt}$, which demonstrates that the optimization based on deterministic equivalent is promising. Also, $\overline{\alpha}^{\rm opt}$ performs well even for \emph{small} system dimension and does provide a significant performance increase. Motivated by the great performance of $\overline{\alpha}^{\rm opt}$, in the sequel, we use it in the discussion for the feedback bit allocation.

\begin{figure}
\centering
\includegraphics[width=0.485\textwidth]{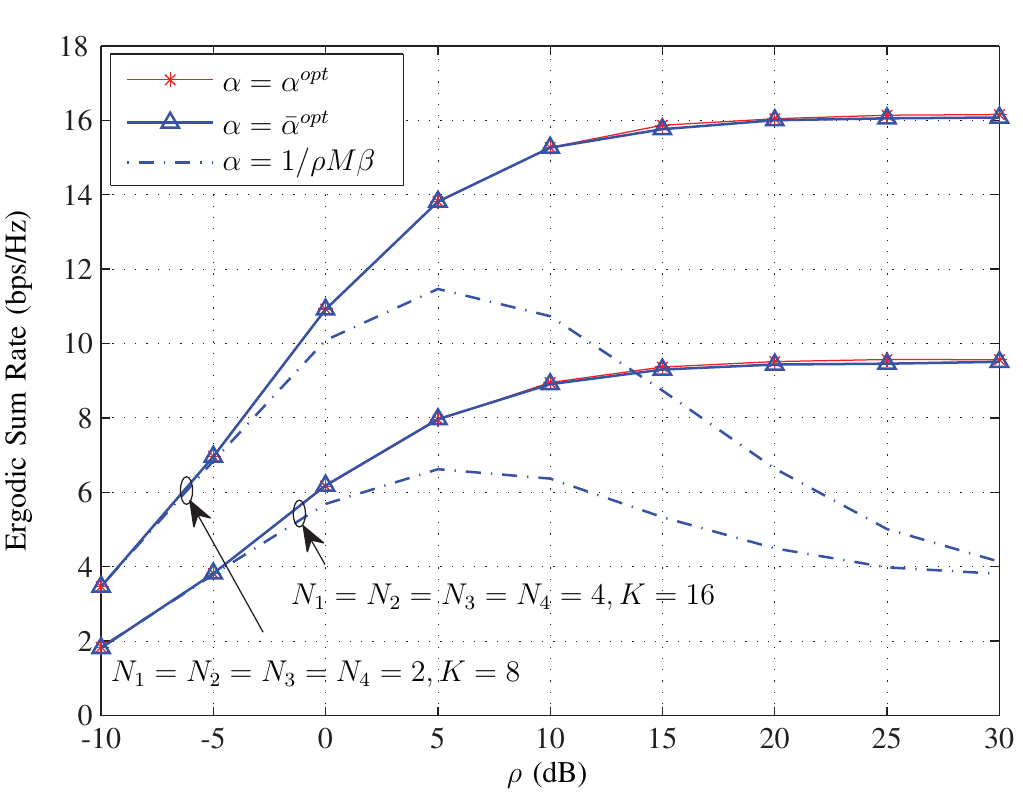}
\caption{Ergodic sum-rate results for various regularization parameters with $\tau_{k,i}^2 = {\rm rand}$ and $\qT \neq \qI$.}\label{fig:optAlpha_MonteVsLarge}
\end{figure}

\subsection{Optimal Feedback Bit Allocation}
In this subsection, we study the optimal feedback bit allocation of each user by maximizing the ergodic sum-rate in FDD systems as shown in Figure~\ref{fig:system model}. Each user quantizes the perfect channel vectors between ${\sf UE}_k$ and ${\sf BS}_i$ using $B_{k,i}$ bits \cite{Jindal06IT,Kerret-11}, and feeds them back to ${\sf BS}_i$ by the
finite-bandwidth feedback channels. We assume that the total number of feedback bits for each user is $B$. Therefore each BS obtains imperfect CSIT and all of this
CSIT can be shared among the BSs for coordinated transmission via the high-speed backhaul channels. The channel model with limited feedback is described by
\eqref{eq:an imperfect estimate H_km}, where $\tau_{k,i}$ denotes the quantization error between ${\sf UE}_k$ and ${\sf BS}_i$ and satisfies $\tau_{k,i}^2 \leq
2^{-\frac{B_{k,i}}{N_i-1}}$ \cite{Jindal06IT}. Hence, if we set $\tau_{k,i}^2 = 2^{-\frac{B_{k,i}}{N_i-1}}$, the optimal feedback bit allocation problem can be
expressed as
\begin{equation} \label{eq:optimalbit1}
\left\{ B_{k,1}^{\rm opt},\dots, B_{k,M}^{\rm opt} \right\}_{\forall k}= \argmax_{B_{k,1},\dots, B_{k,M} \in {\mathbb B}, \forall k} R_{\rm sum},
\end{equation}
where
\begin{align*}
{\mathbb B} \triangleq  \bigg\{ B_{k,i}, & \forall k,i \left| \sum^M_{i=1} B_{k,i} = B,\right.  \nonumber \\
                    & B_{k,i}\mbox{'s are non-negative integers}, \forall k \bigg\}
\end{align*}
is the feasible set of $B_{k,i}$'s. Since the number of candidates in the feasible set is finite, the optimal solution can be obtained by an exhaustive search. However, the required complexity is prohibitive when either $B$ or $M$ is large, due to the need of Monte-Carlo averaging. Instead, we propose to solve
\begin{equation} \label{eq:optimalbit2}
\left\{ \overline{B}_{k,1}^{\rm opt},\dots, \overline{B}_{k,M}^{\rm opt} \right\}_{\forall k}= \argmax_{B_{k,1},\dots, B_{k,M} \in {\mathbb B}, \forall k} \overline{R}_{\rm sum}.
\end{equation}
Note that in \eqref{eq:optimalbit2}, the asymptotic-optimal regularization parameter $\overline{\alpha}^{\rm opt}$ is adopted.

\begin{figure}
\centering
\includegraphics[width=0.485\textwidth]{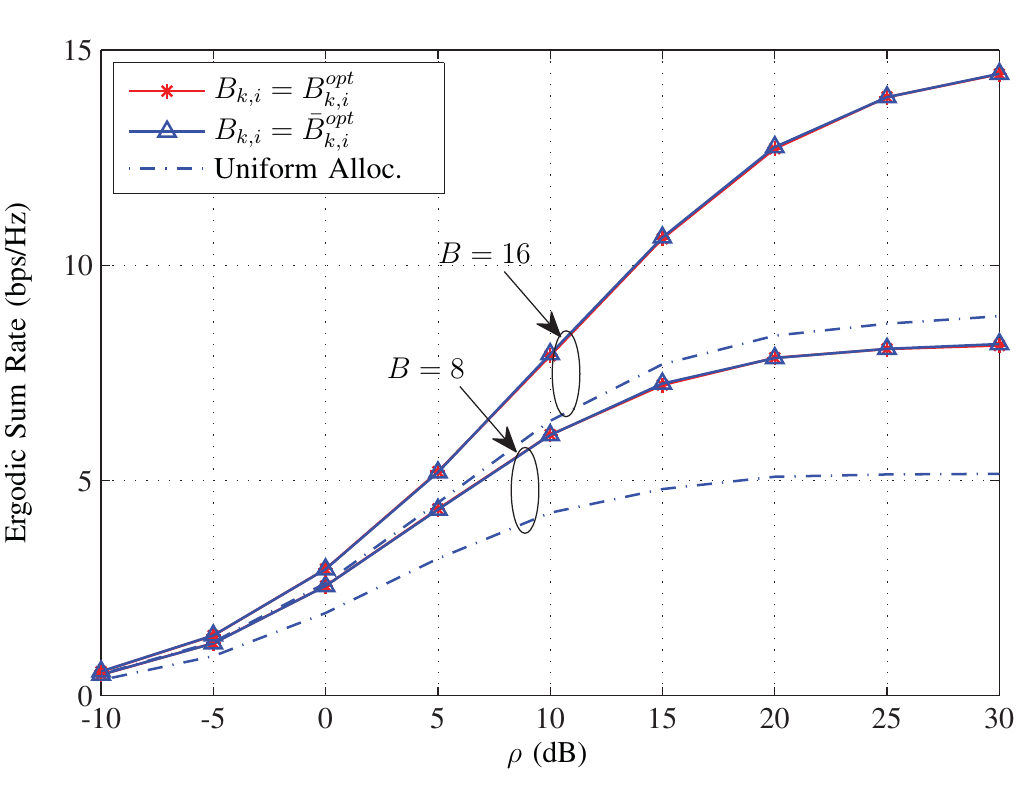}
\caption{Ergodic sum-rate results versus SNR for different bit allocations.}\label{fig:optBit_MonteVsLarge}
\end{figure}

In Fig.~\ref{fig:optBit_MonteVsLarge}, we compare the ergodic sum-rate results for using $B_{k,i}^{\rm opt}$, $\overline{B}_{k,i}^{\rm opt}$, and uniform bit allocation with $M=2$, $N_1=N_2=4$, $K=4$, $\{\qT_{k,i} = \varrho_i \qI\}_{\forall k}$, and $\frac{\varrho_2}{\varrho_1}=0.0125$. Note that as mentioned in Section II of the network model, we have absorbed the channel-path gain into $\qT_{k,i}$. It is observed that using $\{\overline{B}_{k,i}^{\rm opt}\}$ performs close to those with $\{B_{k,i}^{\rm opt}\}$ and the performance is significantly degraded if uniform allocation is adopted. Since the optimization based on the deterministic equivalent is computationally much more efficient and near-optimal, $\{\overline{B}_{k,i}^{\rm opt}\}$ will be regarded as the optimal bit allocation in the following discussions.

It is anticipated that the weaker the channel-path gain between the user and the BS, the less the CSI quantization bits should be allocated. For clarity, we start the discussion from the case with two BSs. In Fig.~\ref{fig:optBit_TotalBitVsFadRatio_Monte}, we show the ergodic sum-rate results against the inter-cell attenuation ratio $\frac{\varrho_1}{\varrho_2}$ when $M=2$, $N_1=N_2=4$, $K=4$, $\rho = 10$ dB, and $\{\qT_{k,i} = \varrho_i \qI\}_{\forall k}$. For the case when $\frac{\varrho_1}{\varrho_2}=1$, the channel-path gains between the two BSs are identical and the optimal bit allocation becomes uniform allocation. As $\frac{\varrho_1}{\varrho_2}$ increases, the performance difference between the optimal and uniform bit allocations increases. As can be seen, when the total number of feedback bits increases, the gain due to the optimal bit allocation decreases. This characteristic is reasonable, since the CSI quality of the stronger link has reached to an acceptable level if the number of feedback bits is very large. This reveals the fact that a proper bit allocation is required when the total
number of feedback bits is at a practical level, e.g., $B = 8$ and $16$.

\begin{figure}
\centering
\includegraphics[width=0.5\textwidth]{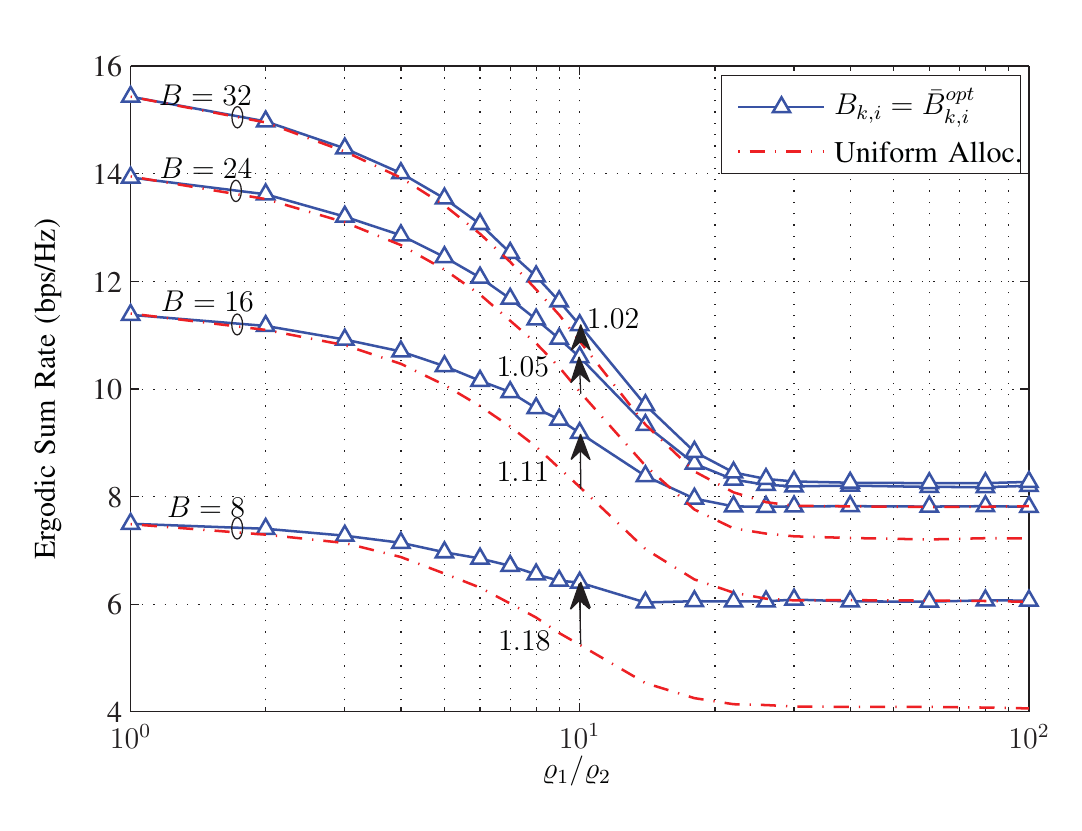}
\caption{The deterministic sum-rate results versus inter-cell ratio $\frac{\varrho_1}{\varrho_2}$ for different bit allocations when $\{ \qT_{k,i} =  \varrho_i \qI\}_{\forall k,i}$ and $\rho=$ 10dB.}\label{fig:optBit_TotalBitVsFadRatio_Monte}
\end{figure}

\begin{figure}
\centering
\includegraphics[width=0.51\textwidth]{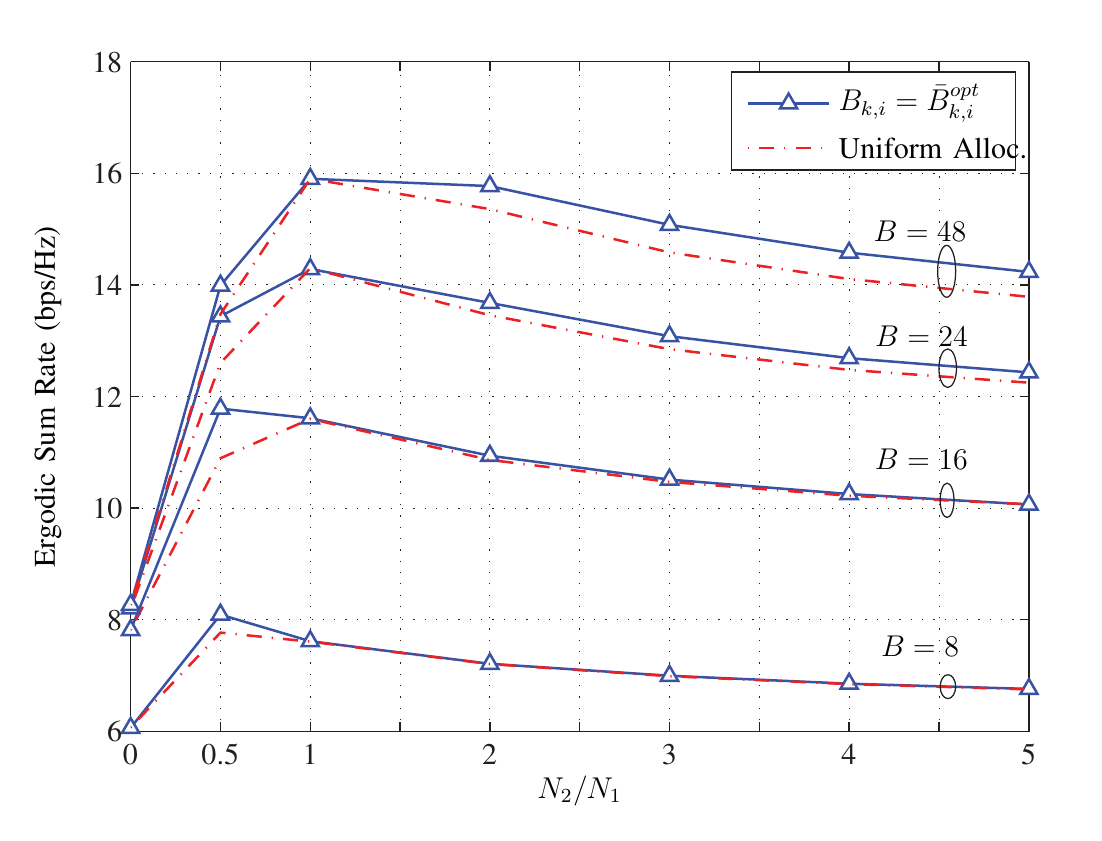}
\caption{Ergodic sum-rate results versus antenna ratio for different bit allocations when $\{ \qT_{k,i} =  \varrho_i \qI\}_{\forall k,i}$ and $\rho=$ 10dB.}\label{fig:optBit_TotalBitVsAntRatio_Monte}
\end{figure}

It is well understood that the BS with more number of antennas provides better performance due to higher inter-user interference mitigation capability. As a result, another important question is that in the case the channel-path gains between the two BSs are identical, whether the BS with more number of antennas requires more feedback bits. In Fig.~\ref{fig:optBit_TotalBitVsAntRatio_Monte}, we compare the ergodic sum-rates between the optimal and uniform bit allocations against the BS antenna ratio $\frac{N_2}{N_1}$ when $M=2$, $N_1=4$, $\rho = 10$ dB, $\{\qT_{k,i} = \varrho_i \qI\}_{\forall k}$, and $\frac{\varrho_2}{\varrho_1}=1$. Results demonstrate that if the total number of feedback bits is at a practical level, uniform bit allocation provides a comparable result to the optimal bit allocation even when $\frac{N_2}{N_1}$ is large. The difference will become more apparent if the number of feedback bits becomes very large in which the more the number of antennas the more the number of feedback bits. Also, it is observed that given a fixed number of feedback bits, the ergodic sum-rates decrease with the BS antenna ratio. For example, when $B = 48$, the sum-rate for the case with $\frac{N_2}{N_1}=2$ is greater than that with $\frac{N_2}{N_1}=1$. Based on these discussions, we conclude that if the channel-path gains are comparable and the total number of feedback bits is at a practical level, performing the optimal bit allocation does not offer much gain in rate.

Optimal bit allocation by exhaustive search is complex if the number of BSs is greater than 2. As discussed, we know that the feedback bit allocation highly depends on the channel-path gain. Specifically, if $\varrho_1 \geq \varrho_2$, we have $B_{k,1} \geq B_{k,2}$ for $\forall k$. With this characteristic, the possible bit combinations can be greatly reduced by introducing an additional restriction. For example, if the channel-path gains are ranked in decreasing order, such as $\varrho_{j_1} \geq \varrho_{j_2} \geq \cdots \geq \varrho_{j_M}$, then the search space can be expressed as
\begin{align}
{\mathbb B}^{\circ} \triangleq & \bigg\{ B_{k,i}, \forall k,i \left| \sum^M_{i=1} B_{k,i} = B, ~B_{k,j_1} \geq B_{k,j_2} \geq \cdots \right.   \nonumber \\
  & \geq  B_{k,j_M}, ~B_{k,i} \mbox{'s are non-negative integers}, \forall k   \bigg\}.  \label{eq:withConstrain}
\end{align}
Clearly, ${\mathbb B}^{\circ} \subset {\mathbb B}$. In Fig.~\ref{fig:SearchSpace}, we show the numbers of candidates in ${\mathbb B}$ and ${\mathbb B}^{\circ}$ as a function of $B$. As can be seen, the search space can be greatly reduced. For example, if $B=9$ and $M=5$ (or $M=3$), the numbers of candidates are reduced from $715$ (or $55$) to $23$ (or $12$), respectively.

\begin{figure}
\centering
\includegraphics[width=0.505\textwidth]{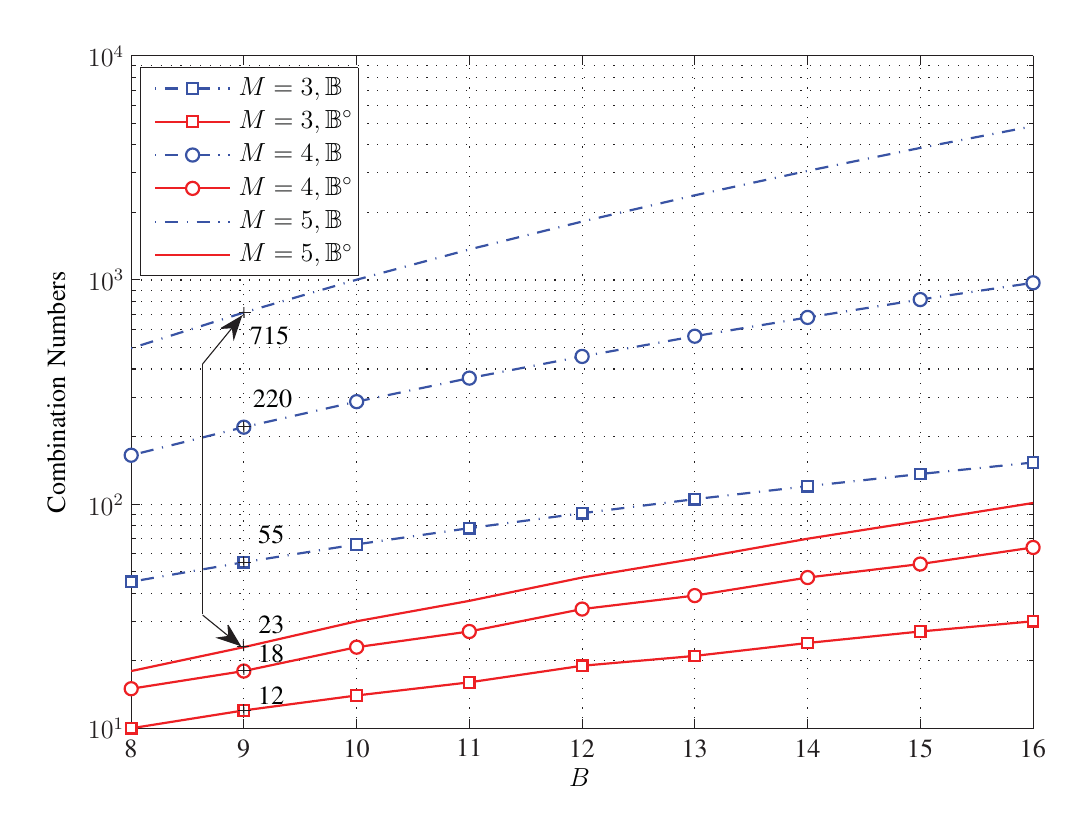}
\caption{The numbers of combinations for the exhaustive search and those with the additional restriction of (\ref{eq:withConstrain}).}\label{fig:SearchSpace}
\end{figure}

\begin{figure}
\centering
\includegraphics[width=0.495\textwidth]{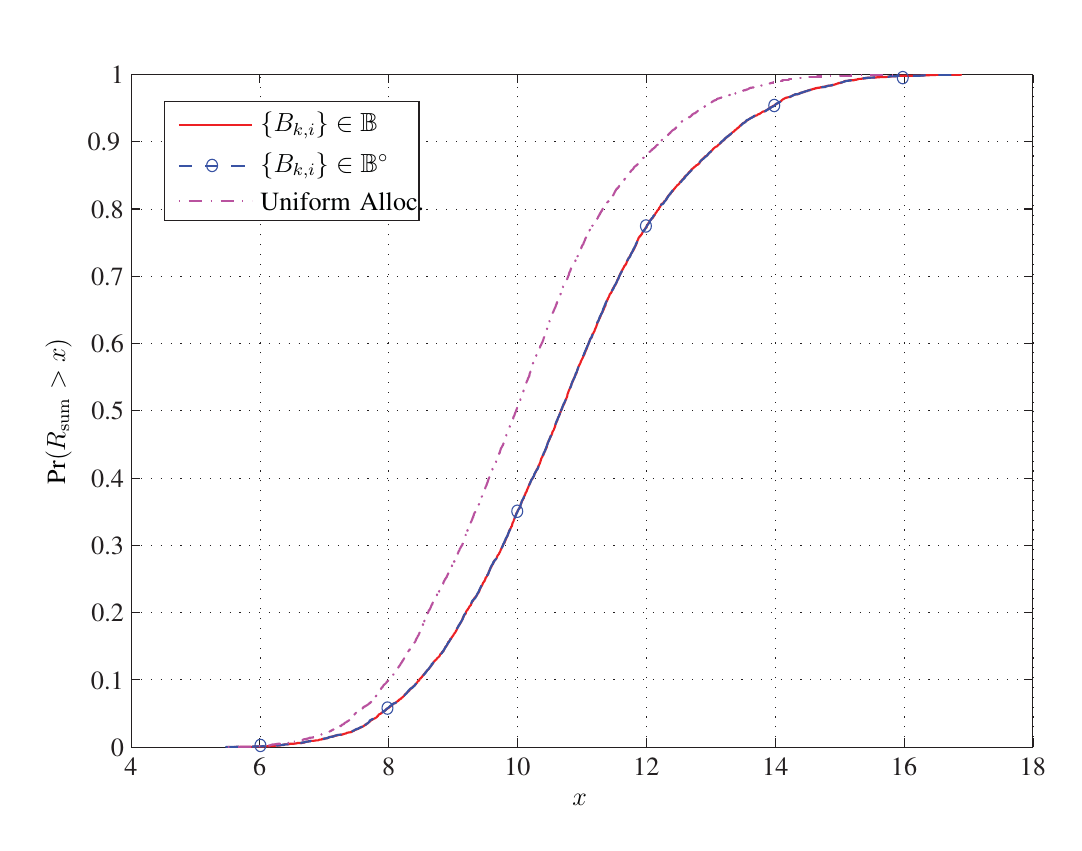}
\caption{The cdf of the ergodic sum-rates for different bit allocations when $M=3$, $N_1=N_2=N_3=3$, $\rho = 20$ dB, $\frac{\varrho_1}{\varrho_2}=-5$, $\frac{\varrho_1}{\varrho_3}= 10$, and $\{\qT_{k,i} \neq \varrho_i \qI\}_{\forall k,i}$. }\label{fig:optBit_CDF}
\end{figure}

Finally, we conclude this section by showing the cumulative distribution function (cdf) results of the ergodic sum-rates when the optimal bit allocations are searched through ${\mathbb B}$ and ${\mathbb B}^{\circ}$ and the results are obtained from $10^4$ independent and arbitrary spatial correlation patterns in Fig.~\ref{fig:optBit_CDF}. Recalling from the discussion in Section III, we regard $\tr \left(\qT_{k,i} \qPsi_i \right)$ as the equivalent channel gain contributed by ${\sf BS}_i$ to ${\sf UE}_k$. Specifically, we rank $ \tr \left(\qT_{k,{j_1}} \qPsi_{j_1} \right) \geq \ldots \geq \tr \left(\qT_{k,{j_M}} \qPsi_{j_M}\right) $ in decreasing order and adopt ${\mathbb B}^{\circ}$ as the feasible set for each user. The search based on ${\mathbb B}^{\circ}$ provides almost identical results to the exhaustive search.

\section{Conclusion}
Using large dimensional RMT, we studied the multi-cell downlink system with cooperative BSs and multiple single-antenna UEs. The deterministic equivalent of the ergodic sum-rate for the RZF system with imperfect CSIT was derived. Simulation results have revealed that the deterministic equivalent sum-rate provides reliable performance predictions even for small system dimensions. Motivated by this, we used the deterministic equivalent result to determine the asymptotic-optimal regularization parameter and the asymptotic-optimal feedback bit allocation. For both applications, we demonstrated that the proposed results achieve indistinguishable results to those obtained by the exhaustive Monte-Carlo method. Several insights and intuitions have been gained for the feedback bit allocation problem and in particularly, we showed that the channel-path gain plays a more important role than the BS antenna ratio.

\appendices

\section{Proof of Theorem \ref{Th: 2}}\label{Appendix: Proof of Theorem 2}

Before proceeding, we first introduce the following theorem which serves as the mathematical basis for the derivation of the multi-cell multiple-input single-output (MISO) channel.

\begin{Theorem}\label{Th: 1}
Consider an $N\times N$ matrix of the following form:
\begin{equation*}
    \qB = \sum_{k=1}^{K} \qT_k^{\frac{1}{2}} \qx_k \qx_k^H \qT_k^{\frac{1}{2}}.
\end{equation*}
In addition to Assumption \ref{Ass:1} in Section III, we suppose that $\qQ \in \bbC^{N\times N}$ is a nonnegative definite matrix with uniformly bounded spectral
norm and $M$ is a finite non-negative integer. Define the matrix product Stieltjes transform of $\qB$ as
\begin{equation*}
    m_{\qB, \qQ}(\alpha)=\frac{1}{N}\tr \left(\qQ(\qB+\alpha\qI_N)^{-1}\right).
\end{equation*}
Then, as $\calN \rightarrow \infty$, we have
\begin{equation*}
    m_{\qB,\qQ}(\alpha) - \overline{m}_{\qB,\qQ} (\alpha) \xrightarrow{a.s.} 0,~~\mbox{for}~ \alpha \in \bbR^+,
\end{equation*}
with $\overline{m}_{\qB,\qQ} (\alpha)$ given by
\begin{equation*}
    \overline{m}_{\qB,\qQ} (\alpha) = \frac{1}{N} \tr \left(\qQ \diag \left( \qPsi_1(\alpha), \qPsi_2(\alpha), \dots , \qPsi_M (\alpha)\right)\right),
\end{equation*}
where
\begin{equation*}
    \qPsi_i (\alpha)=\left(  \frac{1}{N_i} \sum_{k=1}^{K}\frac{1}{\sum_{m=1}^{M} e_{k,m}(\alpha)+1} \qT_{k,i} + \alpha\qI_{N_i} \right)^{-1},
\end{equation*}
and ${e_{k,i}}(\alpha)$'s form the unique solution of the following $K \times M$ equations
\begin{equation*}
    e_{k,i}(\alpha) = \frac{1}{N_i}\tr \left(\qT_{k,i} \qPsi_i (\alpha) \right),
\end{equation*}
for $ k =1,\dots K$, and $ i = 1,\dots,M$.
\end{Theorem}

\begin{proof}
The result can be obtained immediately from \cite[Theorem 1]{Wagner12IT}.\footnote{In this paper, $\qQ \in \bbC^{N\times N}$ is a nonnegative definite matrix, but in \cite[(98)]{Wagner12IT}, $\qQ$ was considered positive definite. However, it can be easily shown that the proof in \cite{Wagner12IT} still works without the positive definite condition on $\qQ$. }
\end{proof}

Note that $m_{\qB,\qQ}(\alpha)$, $\overline{m}_{\qB,\qQ} (\alpha)$, $\qPsi_i (\alpha)$, and $e_{k,i}(\alpha)$ are all functions of $\alpha$ but for ease of notations, $\alpha$ is dropped. In addition, all the following approximations will be done in the limit $\calN \rightarrow \infty$ and for ease of expression, the notation ``$\rightarrow$'' will represent the limit for $\calN \rightarrow \infty$.

Theorem \ref{Th: 1} indicates that $m_{\qB,\qQ}$ can be approximated by its deterministic equivalent $\overline{m}_{\qB,\qQ}$ without knowing the actual realization of $\qx_k$'s. The deterministic equivalent is analytical and is much easier to compute than $\Ex_{\qB}\{ m_{\qB,\qQ} \}$, which requires time-consuming Monte-Carlo simulations. Motivated by this result in the large system limit, we aim to derive the deterministic equivalent of $\gamma_{k}$.

The SINR $\gamma_{k}$ in \eqref{eq:SINR} consists of 1) the noise power $\nu$, 2) the signal power $| \qh_k^H \hat{\qW}\hat{\qh}_k|^2$, and 3) the interference power $\qh_k^H \hat{\qW} \hat{\qH}_{[k]}^H \hat{\qH}_{[k]} \hat{\qW} \qh_k$. We will derive the deterministic equivalent of each term in the following three lemmas. Though the procedure for the MISO channel without multi-cell cooperation \cite{Wagner12IT} is used, several nontrivial manipulations (especially Lemma \ref{Lemma 1}) for the multi-cell system with spatial correlations are required.

\begin{Lemma}\label{Lemma first term}
Under Assumption \ref{Ass:1}, as $\calN \rightarrow \infty$, we have
\begin{equation}\label{eq: asymptotic expression nu}
    \nu-\overline{\nu} \xrightarrow{a.s.} 0,
\end{equation}
where $\overline{\nu}$ is given by \eqref{eq: nu_circ}.
\end{Lemma}

\begin{proof}
According the definition of $\nu$, we have
\begin{equation*}
   \nu  = \mathop {\max }\limits_i  \frac{N}{N_{i}\rho}\left( m_{\hat{\qH}^H \hat{\qH}, \qE_i} - \alpha \dot{m}_{\hat{\qH}^H \hat{\qH}, \qE_i} \right),
\end{equation*}
where $\dot{m}_{\hat{\qH}^H \hat{\qH}, \qE_i}$ is the derivative of $m_{\hat{\qH}^H \hat{\qH}, \qE_i}$ w.r.t.~$\alpha$. Using Theorem \ref{Th: 1}, we get
\begin{equation}\label{eq: asymptotic expression m}
    m_{\hat{\qH}^H \hat{\qH}, \qE_i}-\overline{m}_{\hat{\qH}^H \hat{\qH}, \qE_i} \xrightarrow{a.s.} 0,
\end{equation}
where
\begin{equation}\label{eq: mEi}
    \overline{m}_{{{{\bf{\hat H}}}^H}{\bf{\hat H}},{{\bf{E}}_i}}  = \frac{1}{N}\tr{\bf \Psi}_i .
\end{equation}
Since the derivative of $\overline{m}_{\hat{\qH}^H \hat{\qH}, \qE_i}$ is a deterministic equivalent for $\dot{m}_{\hat{\qH}^H \hat{\qH}, \qE_i}$, we have
\begin{equation}\label{eq: asymptotic expression dotm}
    \dot{m}_{\hat{\qH}^H \hat{\qH}, \qE_i}-\overline{\dot m}_{\hat{\qH}^H \hat{\qH}, \qE_i} \xrightarrow{a.s.} 0.
\end{equation}
From \eqref{eq: asymptotic expression m} and \eqref{eq: asymptotic expression dotm}, we get \eqref{eq: asymptotic expression nu},
where
\begin{equation}\label{eq: nu_circ1}
    \overline{\nu} = \mathop {\max }\limits_i \frac{N}{N_{i}\rho}\left(\overline{m}_{\hat{\qH}^H \hat{\qH}, \qE_i} - \alpha \overline{\dot m}_{\hat{\qH}^H \hat{\qH}, \qE_i}\right).
\end{equation}
The first term of $\overline{\nu}$ can be calculated by \eqref{eq: mEi}. Next, we proceed to derive the second term. Taking the derivative of
$\overline{m}_{\hat{\qH}^H \hat{\qH}, \qE_i}$ w.r.t.~$\alpha$, we have
\begin{equation}\label{eq: dotm}
\overline{\dot m}_{\hat{\qH}^H \hat{\qH}, \qE_i} = \frac{1}{N}\tr {\bf \Psi}_i^2 - \frac{1}{N} \sum _{k=1}^{K} \dot{c}_{k,i}  \tr \left(\qPsi_i \qT_{k,i} \qPsi_i\right),
\end{equation}
where
\begin{align}
\dot{c}_{k,i}  &= \frac{1}{N_i} \frac{-\sum_{j = 1}^M \dot{e}_{k,j}}{\left(\sum_{j = 1}^M e_{k,j}+1\right)^2}, \label{eq: dotc_k} \\
\dot{e}_{k,j} &= \frac{1}{N_j}\tr \left(\qT_{k,j} \qPsi^2_j \right) - \frac{1}{N_j} \sum _{l=1}^{K} \dot{c}_{l,j}  \tr \left(\qT_{k,j} \qPsi_j \qT_{l,j} \qPsi_j \right), \label{eq: dote_ki}
\end{align}
and $\qPsi_i$ is given by \eqref{eq: Psi i}. Plugging \eqref{eq: dote_ki} into \eqref{eq: dotc_k}, we have that $\dot{\qC}   = [\dot{c}_{k,i}  ] \in \bbC^{K \times
M}$ is a solution of the linear equation \eqref{eq: Theta}. Substituting \eqref{eq: mEi} and \eqref{eq: dotm} into \eqref{eq: nu_circ1}, $\overline{\nu}$ is
explicitly given by \eqref{eq: nu_circ}.
\end{proof}

\begin{Lemma}\label{Lemma second term}
Under Assumption \ref{Ass:1}, as $\calN \rightarrow \infty$, we have
\begin{equation}\label{eq: asymptotic expression signal power}
\qh_k^H \hat{\qW} \hat{\qh}_k - \frac{\overline{u}^{(2)}_k}{1 + \overline{u}^{(1)}_k } \xrightarrow{a.s.} 0,
\end{equation}
where $\overline{u}^{(1)}_k$ and $\overline{u}^{(2)}_k$ are given by \eqref{eq: u1_circ} and \eqref{eq: u2_circ}, respectively.
\end{Lemma}

\begin{proof}
For the signal power $| \qh_k^H \hat{\qW}\hat{\qh}_k|^2$, we have
\begin{align}\label{eq: signal power}
\qh_k^H \hat{\qW} \hat{\qh}_k &= \frac{1}{1 + \hat{\qh}_k^H \qA_{[k]}^{-1} \hat{\qh}_k} \qh_k^H \qA_{[k]}^{-1} \hat{\qh}_k  \nonumber\\
 &= \frac{1}{1 + \hat{\qx}_k^H \qT^{\frac{1}{2}}_k \qA_{[k]}^{-1} \qT^{\frac{1}{2}}_k \hat{\qx}_k} \qx_k^H \qT^{\frac{1}{2}}_k \qA_{[k]}^{-1} \qT^{\frac{1}{2}}_k  \qLambda_k \qx_k  \nonumber\\
 &\quad + \frac{1}{1+ \hat{\qx}_k^H \qT^{\frac{1}{2}}_k \qA_{[k]}^{-1} \qT^{\frac{1}{2}}_k \hat{\qx}_k} \qx_k^H \qT^{\frac{1}{2}}_k \qA_{[k]}^{-1} \qT^{\frac{1}{2}}_k  \qOmega_k \qv_k,
\end{align}
where $\qLambda_k = \diag \left(\psi_{k,1} \qI_{N_1}, \ldots, \psi_{k,M} \qI_{N_M} \right)$, $\qOmega_k = \diag \left(\tau_{k,1} \qI_{N_1}, \ldots, \tau_{k,M} \qI_{N_M} \right)$, and $\qA_{[k]} \triangleq \hat{\qH}_{[k]}^H \hat{\qH}_{[k]} + \alpha \qI_N$, the first equality follows by the matrix inverse lemma
\cite[Lemma 2.1]{Bai-09},\footnote{\cite[Lemma 2.1]{Bai-09}: For any $\qA\in \bbC^{n\times n}$ and $\qq \in \bbC^n$ with $\qA$ and $\qA+\qq\qq^H$ invertible, we
have $$\qq^H \left( \qA+\qq\qq^H \right)^{-1} = \frac{1}{1+\qq^H\qA^{-1}\qq} \qq^H\qA^{-1}.$$} and the second equality follows merely from the definitions of $\qh_k$
and $\hat{\qh}_k$ in \eqref{eq:the correlated channel model} and \eqref{eq:an imperfect estimate H_km}. Since $\qx_k$ and $\qv_k$ are independent, $\qx_k^H
\qT^{\frac{1}{2}}_k \qA_{[k]}^{-1} \qT^{\frac{1}{2}}_k \qOmega_k \qv_k$ almost surely converges to zero. Also, applying \cite[Lemma 2.3]{Bai-09}, we have that
$\hat{\qx}_k^H \qT^{\frac{1}{2}}_k \qA_{[k]}^{-1} \qT^{\frac{1}{2}}_k \hat{\qx}_k$ almost surely converges to $u^{(1)}_k =\frac{1}{N}\tr (\qT_k \qA_{[k]}^{-1})$
and $\qx_k^H \qT^{\frac{1}{2}}_k \qA_{[k]}^{-1} \qT^{\frac{1}{2}}_k \qLambda_k \qx_k$ almost surely converges to $u^{(2)}_k =\frac{1}{N}\tr (\qT^{\frac{1}{2}}_k
\qLambda_k \qT^{\frac{1}{2}}_k \qA_{[k]}^{-1})$. Let $\qA = \hat{\qH}^H \hat{\qH} + \alpha \qI_M$. Using the fact \cite[Lemma 2.2]{Bai-09} that $
\frac{1}{N}\tr(\qQ \qA_{[k]}^{-1}) -\frac{1}{N}\tr(\qQ \qA^{-1}) \xrightarrow{a.s.} 0$ and Theorem \ref{Th: 1}, we have
\begin{equation}\label{eqn: u1u20}
 u^{(1)}_k - \overline{u}^{(1)}_k  \xrightarrow{a.s.} 0,~~\mbox{and}~~u^{(2)}_k - \overline{u}^{(2)}_k  \xrightarrow{a.s.} 0,
\end{equation}
where $\overline{u}^{(1)}_k = \overline{m}_{\hat{\qH}^H \hat{\qH}, \qT_k} $ and $\overline{u}^{(2)}_k = \overline{m}_{\hat{\qH}^H
\hat{\qH},\qT^{\frac{1}{2}}_k\qLambda_k \qT^{\frac{1}{2}}_k} $. Here, $\overline{u}^{(1)}_k$ and $\overline{u}^{(2)}_k$ can be obtained by Theorem \ref{Th: 1} and
are explicitly shown in \eqref{eq: u1_circ} and \eqref{eq: u2_circ} respectively. From \eqref{eq: signal power} and (\ref{eqn: u1u20}), we have \eqref{eq: asymptotic expression signal power}.
\end{proof}

\begin{Lemma}\label{Lemma third term}
Under Assumption \ref{Ass:1}, as $\calN \rightarrow \infty$, we have
\begin{equation}\label{eq: asymptotic expression interference power}
    \qh_k^H \hat{\qW} \hat{\qH}_{[k]}^H \hat{\qH}_{[k]} \hat{\qW} \qh_k - \overline{u}_k \xrightarrow{a.s.} 0,
\end{equation}
where $\overline{u}_k$ has been obtained by \eqref{eq: u_circ}.
\end{Lemma}

\begin{proof}
Consider the interference power $\qh_k^H \hat{\qW} \hat{\qH}_{[k]}^H \hat{\qH}_{[k]} \hat{\qW} \qh_k$ by writing it as
\begin{equation}\label{eq: interference power 1}
    \qh_k^H \hat{\qW} \hat{\qH}_{[k]}^H \hat{\qH}_{[k]} \hat{\qW} \qh_k = \qx_k^H \qT_k^{\frac{1}{2}} \qA^{-1} \hat{\qH}_{[k]}^H \hat{\qH}_{[k]} \qA^{-1} \qT_k^{\frac{1}{2}} \qx_k.
\end{equation}
In order to eliminate the dependence between $\qx_k$ and $\qA$ in \eqref{eq: interference power 1}, recall the definition:
\begin{align}\label{eq: A}
    \qA &= \hat{\qH}^H \hat{\qH} + \alpha \qI_M \nonumber\\
        &= \qA_{[k]} + \qT_k^{\frac{1}{2}} \biggl( \qLambda_k \qx_k \qx_k^H \qLambda_k + \qOmega_k \qv_k \qv_k^H \qOmega_k  \nonumber\\
         &\quad + \qLambda_k \qx_k \qv_k^H \qOmega_k + \qOmega_k \qv_k \qx_k^H \qLambda_k \biggr) \qT_k^{\frac{1}{2}}.
\end{align}
With the above, we can rewrite \eqref{eq: interference power 1} as \eqref{eq: interference power 2}, shown at the top of next page.
\begin{figure*}[!t]
\begin{align} \label{eq: interference power 2}
 \qh_k^H \hat{\qW} \hat{\qH}_{[k]}^H \hat{\qH}_{[k]} \hat{\qW} \qh_k =& \qx_k^H \qT_k^{\frac{1}{2}} \qA_{[k]}^{-1} \hat{\qH}_{[k]}^H \hat{\qH}_{[k]} \qA^{-1} \qT_k^{\frac{1}{2}} \qx_k + \qx_k^H \qT_k^{\frac{1}{2}} \left( \qA^{-1} - \qA_{[k]}^{-1} \right) \hat{\qH}_{[k]}^H \hat{\qH}_{[k]} \qA^{-1} \qT_k^{\frac{1}{2}} \qx_k \nonumber\\
  =& \qx_k^H \qT_k^{\frac{1}{2}} \qA^{-1} \qT_k^{\frac{1}{2}} \qx_k - \alpha \qx_k^H \qT_k^{\frac{1}{2}} \qA_{[k]}^{-1} \qA^{-1} \qT_k^{\frac{1}{2}} \qx_k \nonumber\\
   &  - \qx_k^H \qT_k^{\frac{1}{2}} \qA^{-1} \qT_k^{\frac{1}{2}} \qLambda_k \qx_k \left( \qx_k^H \qLambda_k \qT_k^{\frac{1}{2}} \qA^{-1} \qT_k^{\frac{1}{2}} \qx_k - \alpha \qx_k^H \qLambda_k \qT_k^{\frac{1}{2}} \qA_{[k]}^{-1} \qA^{-1} \qT_k^{\frac{1}{2}}\qx_k \right)  \nonumber\\
   &  - \qx_k^H \qT_k^{\frac{1}{2}} \qA^{-1} \qT_k^{\frac{1}{2}} \qOmega_k \qv_k  \left( \qv_k^H \qOmega_k  \qT_k^{\frac{1}{2}} \qA^{-1} \qT_k^{\frac{1}{2}} \qx_k - \alpha \qv_k^H \qOmega_k  \qT_k^{\frac{1}{2}} \qA_{[k]}^{-1} \qA^{-1} \qT_k^{\frac{1}{2}}\qx_k \right)  \nonumber\\
   &  - \qx_k^H \qT_k^{\frac{1}{2}} \qA^{-1} \qT_k^{\frac{1}{2}} \qLambda_k \qx_k \left( \qv_k^H \qOmega_k  \qT_k^{\frac{1}{2}} \qA^{-1} \qT_k^{\frac{1}{2}} \qx_k - \alpha \qv_k^H \qOmega_k  \qT_k^{\frac{1}{2}} \qA_{[k]}^{-1} \qA^{-1} \qT_k^{\frac{1}{2}}\qx_k \right)  \nonumber\\
   &  - \qx_k^H \qT_k^{\frac{1}{2}} \qA^{-1} \qT_k^{\frac{1}{2}} \qOmega_k \qv_k  \left( \qx_k^H \qLambda_k \qT_k^{\frac{1}{2}} \qA^{-1} \qT_k^{\frac{1}{2}} \qx_k - \alpha \qx_k^H \qLambda_k \qT_k^{\frac{1}{2}} \qA_{[k]}^{-1} \qA^{-1} \qT_k^{\frac{1}{2}}\qx_k \right)
\end{align}
\hrulefill
\end{figure*}
To calculate the deterministic equivalent of each term above, the following Lemma \ref{Lemma 1} extends \cite[Lemma 7]{Wagner12IT} to deal with the deterministic equivalent of the considered
multi-cell system with different CSIT qualities even inside a channel vector from a user to all BSs.

%%%%%----------Lemma1---------
\begin{Lemma}\label{Lemma 1}
Let $\qA, \qD \in \bbC^{N\times N}$ be invertible matrices and $\qU, \qV, \qLambda, \qOmega \in \bbC^{N\times N}$ be of uniformly bounded spectral norm, and satisfy
\begin{equation*}
\qA = \qD +  \qLambda \qx \qx^H \qLambda^H + \qOmega \qv \qv^H \qOmega^H + \qLambda \qx \qv^H \qOmega^H + \qOmega \qv \qx^H \qLambda^H,
\end{equation*}
where $\qx, \qv \in \bbC^{N}$ have i.i.d.~zero-mean entries of variance of $\frac{1}{N}$ and finite 8-th order moment and are mutually independent as well as independent of $\qU, \qV, \qD$. Then we have, almost surely,
\begin{align}
    \qx^H \qU & \qA^{-1} \qV \qx - \biggl( \frac{1}{N} \tr \qV \qU \qD^{-1} \nonumber\\
    &-\frac{\frac{1}{N} \tr \qLambda \qU \qD^{-1} \frac{1}{N} \tr \qV \qLambda^H \qD^{-1} }{1 + \frac{1}{N} \tr \qLambda \qLambda^H \qD^{-1} + \frac{1}{N} \tr \qOmega\qOmega^H \qD^{-1} }\biggr) \xrightarrow{a.s.} 0,\label{eq: Lemma 1 results 1}
\end{align}
and
\begin{align}
    \qx^H \qU &\qA^{-1} \qV \qv  \nonumber\\
    &-\frac{- \frac{1}{N} \tr \qLambda \qU \qD^{-1} \frac{1}{N} \tr \qV \qOmega^H \qD^{-1} }{1 + \frac{1}{N} \tr \qLambda \qLambda^H \qD^{-1} + \frac{1}{N} \tr \qOmega \qOmega^H \qD^{-1} }   \xrightarrow{a.s.} 0.\label{eq: Lemma 1 results 2}
\end{align}
\end{Lemma}

\begin{proof}
Using the fact that $\qA^{-1}-\qD^{-1}=-\qA^{-1}(\qA-\qD)\qD^{-1}$, we have %\eqref{eq: xx} and \eqref{eq: xv}, shown at the top of next page.
\begin{align}\label{eq: xx}
     & \qx^H \qU \qA^{-1} \qV \qx - \qx^H \qU \qD^{-1} \qV \qx \nonumber \\
    =& -\qx^H \qU \qA^{-1} (\qA-\qD) \qD^{-1} \qV \qx  \nonumber \\
    =& -\qx^H \qU \qA^{-1} \bigl(\qLambda \qx \qx^H \qLambda^H + \qOmega \qv \qv^H \qOmega^H + \qLambda \qx \qv^H \qOmega^H \nonumber \\
      &+ \qOmega \qv \qx^H \qLambda^H \bigr) \qD^{-1} \qV \qx \nonumber \\
    =& -\bigl(\qx^H \qU \qA^{-1} \qLambda \qx + \qx^H \qU \qA^{-1} \qOmega \qv\bigr) \bigl(\qx^H \qLambda^H \qD^{-1} \qV \qx \nonumber \\
    &+ \qv^H \qOmega^H \qD^{-1} \qV \qx \bigr),
\end{align}
and
\begin{align}\label{eq: xv}
    & \qx^H \qU \qA^{-1} \qV \qv - \qx^H \qU \qD^{-1} \qV \qv \nonumber \\
  = & -\bigl(\qx^H \qU \qA^{-1} \qLambda \qx + \qx^H \qU \qA^{-1} \qOmega \qv\bigr) \bigl(\qx^H \qLambda^H \qD^{-1} \qV \qv \nonumber \\
    & + \qv^H \qOmega^H \qD^{-1} \qV \qv\bigr).
\end{align}
Applying the same method to $\qx^H \qU \qA^{-1} \qLambda \qx$ and $\qx^H \qU \qA^{-1} \qOmega \qv$, we have %\eqref{eq: xx2} and \eqref{eq: xv2}, shown at the top of next page.
\begin{align}\label{eq: xx2}
    & \qx^H \qU \qA^{-1} \qLambda \qx - \qx^H \qU \qD^{-1} \qLambda \qx \nonumber \\
  = & -\bigl(\qx^H \qU \qA^{-1} \qLambda \qx + \qx^H \qU \qA^{-1} \qOmega \qv\bigr) \bigl(\qx^H \qLambda^H \qD^{-1} \qLambda \qx \nonumber \\
    &+ \qv^H \qOmega^H \qD^{-1} \qLambda \qx \bigr),
\end{align}
and
\begin{align}\label{eq: xv2}
    & \qx^H \qU \qA^{-1}  \qOmega \qv - \qx^H \qU \qD^{-1} \qOmega \qv \nonumber \\
   =& -\bigl(\qx^H \qU \qA^{-1} \qLambda \qx + \qx^H \qU \qA^{-1} \qOmega \qv \bigr) \bigl(\qx^H \qLambda^H \qD^{-1} \qOmega \qv \nonumber \\
    & + \qv^H \qOmega^H \qD^{-1} \qOmega \qv \bigr).
\end{align}
From \eqref{eq: xx2} and \eqref{eq: xv2}, we can derive the solutions of $\qx^H \qU \qA^{-1} \qLambda \qx$ and $\qx^H \qU \qA^{-1} \qOmega \qv$. Thus, we have
\begin{align}
    \qx^H \qU &\qA^{-1} \qLambda \qx \nonumber \\
    &- \frac{\frac{1}{N} \tr \qLambda \qU \qD^{-1}\left( 1 + \frac{1}{N} \tr \qOmega \qOmega^H \qD^{-1} \right)} {1 + \frac{1}{N} \tr \qLambda \qLambda^H \qD^{-1} + \frac{1}{N} \tr \qOmega \qOmega^H \qD^{-1} } \xrightarrow{a.s.} 0,\label{eq: xx3}
\end{align}
\begin{align}
    \qx^H \qU &\qA^{-1} \qOmega \qv \nonumber \\
    &- \frac{- \frac{1}{N} \tr \qLambda \qU \qD^{-1} \frac{1}{N} \tr \qOmega \qOmega^H \qD^{-1} }{1 + \frac{1}{N} \tr \qLambda \qLambda^H \qD^{-1} + \frac{1}{N} \tr\qOmega \qOmega^H \qD^{-1} } \xrightarrow{a.s.} 0.\label{eq: xv3}
\end{align}
Substituting \eqref{eq: xx3} and \eqref{eq: xv3} into \eqref{eq: xx} and \eqref{eq: xv}, we obtain \eqref{eq: Lemma 1 results 1} and \eqref{eq: Lemma 1 results 2}.
\end{proof}
%%%%%--------end--Lemma1---------

Now, using Lemma \ref{Lemma 1}, we can easily get
%\begin{subequations}
\begin{align}
   \qx_k^H \qT_k^{\frac{1}{2}} \qA^{-1} \qT_k^{\frac{1}{2}} \qx_k - \left( u^{(1)}_k - \frac{\left(u^{(2)}_k \right)^2}{1+u^{(3)}_k+u^{(4)}_k} \right)
   &\xrightarrow{a.s.} 0,\label{eq: xTATx} \\
   \qx_k^H \qT_k^{\frac{1}{2}} \qA_{[k]}^{-1} \qA^{-1} \qT_k^{\frac{1}{2}} \qx_k - \left( \dot{u}^{(1)}_k - \frac{u^{(2)}_k \dot{u}^{(2)}_k}{1+u^{(3)}_k+u^{(4)}_k} \right)
   &\xrightarrow{a.s.} 0,\label{eq: xTAkATX} \\
    \qx_k^H \qT_k^{\frac{1}{2}} \qA^{-1} \qT_k^{\frac{1}{2}} \qLambda_k \qx_k - \frac{u^{(2)}_k (1+u^{(4)}_k)}{1+u^{(3)}_k+u^{(4)}_k}
    &\xrightarrow{a.s.} 0,\label{eq: xTATLx} \\
    \qx_k^H \qT_k^{\frac{1}{2}} \qA^{-1} \qT_k^{\frac{1}{2}} \qOmega_k \qv_k - \frac{- u^{(2)}_k u^{(4)}_k}{1+u^{(3)}_k+u^{(4)}_k}
    &\xrightarrow{a.s.} 0,\label{eq: xTATOv}  \\
    \qx_k^H \qLambda _k \qT_k^{\frac{1}{2}} \qA_{[k]}^{-1} \qA^{-1} \qT_k^{\frac{1}{2}} \qx_k ~~~~~~~~~~~~~~~~~~~~~~~~~~~~~& \nonumber \\
     - \left( \dot{u}^{(2)}_k - \frac{u^{(2)}_k \dot{u}^{(3)}_k}{1+u^{(3)}_k+u^{(4)}_k} \right)
    &\xrightarrow{a.s.} 0,\label{eq: xLTAkATx} \\
    \qv_k^H \qOmega_k \qT_k^{\frac{1}{2}} \qA_{[k]}^{-1} \qA^{-1} \qT_k^{\frac{1}{2}} \qx_k - \frac{- u^{(2)}_k \dot{u}^{(4)}_k}{1+u^{(3)}_k+u^{(4)}_k}
    &\xrightarrow{a.s.} 0,\label{eq: vOTAkATx}
\end{align}
%\end{subequations}
where $u^{(3)}_k = \frac{1}{N}\tr ( \qT^{\frac{1}{2}}_k \qLambda_k^2  \qT^{\frac{1}{2}}_k \qA_{[k]}^{-1})$,
 $u^{(4)}_k = \frac{1}{N}\tr ( \qT^{\frac{1}{2}}_k \qOmega_k^2  \qT^{\frac{1}{2}}_k \qA_{[k]}^{-1})$,
 $\dot{u}^{(1)}_k = \frac{1}{N}\tr ( \qT_k \qA_{[k]}^{-2})$,
 $\dot{u}^{(2)}_k = \frac{1}{N}\tr ( \qT^{\frac{1}{2}}_k \qLambda_k \qT^{\frac{1}{2}}_k \qA_{[k]}^{-2})$,
 $\dot{u}^{(3)}_k = \frac{1}{N}\tr ( \qT^{\frac{1}{2}}_k \qLambda_k^2 \qT^{\frac{1}{2}}_k \qA_{[k]}^{-2})$, and
 $\dot{u}^{(4)}_k = \frac{1}{N}\tr ( \qT^{\frac{1}{2}}_k \qOmega_k^2 \qT^{\frac{1}{2}}_k \qA_{[k]}^{-2})$.
Substituting \eqref{eq: xTATx}--\eqref{eq: vOTAkATx} into \eqref{eq: interference power 2}, using again the fact that $\frac{1}{N}\tr (\qQ \qA_{[k]}^{-1} )
-\frac{1}{N}\tr(\qQ \qA^{-1}) \xrightarrow{a.s.} 0$, and Theorem \ref{Th: 1}, we obtain \eqref{eq: asymptotic expression interference power}.
\end{proof}

According Lemmas \ref{Lemma first term}, \ref{Lemma second term} and \ref{Lemma third term}, we obtain the deterministic equivalent $\overline{\gamma}_{k}$ of $\gamma_{k}$ in \eqref{eq:gamma deterministic equivalent}.

\section*{Acknowledgment}
We thank the reviewers and the associate editor for the careful reviews and for their suggestions which helped in improving the quality of the paper.

%{\renewcommand{\baselinestretch}{1.1}
\bibliographystyle{IEEEtran}

\end{document}